\documentclass[12pt]{amsart}
\usepackage{amsmath}
\usepackage{amsfonts}
\usepackage{amssymb}
\usepackage[all,cmtip]{xy}           
\usepackage{bm}
\usepackage{bbm}
\usepackage{bbding}
\usepackage{txfonts}
\usepackage{amscd}
\usepackage{xspace}
\usepackage[shortlabels]{enumitem}
\usepackage{ifpdf}

\ifpdf
  \usepackage[colorlinks,final,backref=page,hyperindex]{hyperref}
\else
  \usepackage[colorlinks,final,backref=page,hyperindex,hypertex]{hyperref}
\fi
\usepackage{tikz}
\usepackage[active]{srcltx}

\usepackage{tikz-cd}
\usepackage{tikz}

\makeatletter

\newtheorem{thm}{Theorem}[section]
\newtheorem{lem}[thm]{Lemma}
\newtheorem{cor}[thm]{Corollary}
\newtheorem{pro}[thm]{Proposition}
\theoremstyle{definition}
\newtheorem{ex}[thm]{Example}
\newtheorem{rmk}[thm]{Remark}
\newtheorem{defi}[thm]{Definition}

\newcommand{\nc}{\newcommand}
\newcommand{\delete}[1]{}

\nc{\mlabel}[1]{\label{#1}}  
\nc{\mcite}[1]{\cite{#1}}  
\nc{\mref}[1]{\ref{#1}}  
\nc{\meqref}[1]{\eqref{#1}}  
\nc{\mbibitem}[1]{\bibitem{#1}} 

\delete{
\nc{\mlabel}[1]{\label{#1}{\hfill \hspace{1cm}{\bf{{\ }\hfill(#1)}}}}
\nc{\mcite}[1]{\cite{#1}{{\bf{{\ }(#1)}}}}  
\nc{\mref}[1]{\ref{#1}{{\bf{{\ }(#1)}}}}  
\nc{\meqref}[1]{\eqref{#1}{{\bf{{\ }(#1)}}}}  
\nc{\mbibitem}[1]{\bibitem[\bf #1]{#1}} 
}

\newcommand {\emptycomment}[1]{}
\newcommand {\yh}[1]{{\textcolor{purple}{yh: #1}}}

\newcommand{\hl}[1]{\textcolor{blue}{hl: #1}}

\delete{
\newcommand {\yh}[1]{{yh: #1}}

\newcommand{\hl}[1]{{hl: #1}}

}

\nc{\oprn}{\theta}

\delete{
\setlength{\baselineskip}{1.8\baselineskip}

\newcommand{\emptycomment}[1]{}
\newcommand{\yh}[1]{{\marginpar{*}\scriptsize\textcolor{purple}{yh: #1}}}

}

\nc{\calo}{\mathcal{O}}
\nc{\oop}{$\mathcal{O}$-operator\xspace}
\nc{\oops}{$\mathcal{O}$-operators\xspace}
\nc{\mrho}{{\bm{\varrho}}}
\nc{\bfk}{\mathbf{K}}
\nc{\invlim}{\displaystyle{\lim_{\longleftarrow}}\,}
\nc{\ot}{\otimes}

\nc{\eval}[1]{\Big|_{#1}}

\newcommand{\add}{\frka\frkd}

\newcommand{\lon }{\,\rightarrow\,}
\newcommand{\be }{\begin{equation}}
\newcommand{\ee }{\end{equation}}

\newcommand{\g}{\mathfrak g}
\newcommand{\h}{\mathfrak h}

\nc{\RR}{\mathbb{R}}

\nc{\CC}{\mathbb{C}}

\newcommand{\huaS}{\mathcal{S}}
\newcommand{\huaL}{\mathcal{L}}
\newcommand{\huaR}{\mathcal{R}}

\newcommand{\frka}{\mathfrak a}

\newcommand{\frkd}{\mathfrak d}

\newcommand{\frkk}{\mathfrak k}
\newcommand{\frkl}{\mathfrak l}

\newcommand{\frkp}{\mathfrak p}

\newcommand{\iii}{{\mathbf{1}}}

\newcommand{\br}[1]{   [ \cdot,    \cdot  ]   }

\newcommand{\Der}{\mathrm{Der}}

\newcommand{\Ad}{\mathrm{Ad}}
\newcommand{\Aut}{\mathrm{Aut}}

\newcommand{\gl}{\mathfrak {gl}}

\newcommand{\End}{\mathrm{End}}
\newcommand{\ad}{\mathrm{ad}}

\nc{\CV}{\mathbf{C}}

\begin{document}

\title[Reflections on Rota-Baxter Lie algebras]{Reflections on Rota-Baxter Lie algebras, the classical reflection equation and Poisson homogeneous spaces}

\author{Honglei Lang}
\address{College of Science, China Agricultural University, Beijing, 100083, China
}
\email{hllang@cau.edu.cn}

\author{Yunhe Sheng}
\address{Department of Mathematics, Jilin University, Changchun 130012, Jilin, China}
\email{shengyh@jlu.edu.cn}


\begin{abstract}
In this paper, first we introduce the notion of reflections on quadratic Rota-Baxter Lie algebras of weight $\lambda$,  and show that they give rise to solutions of the classical reflection equation for the corresponding triangular Lie bialgebra ($\lambda=0$) and factorizable Lie bialgebra ($\lambda\neq0$). Then we study reflections on relative Rota-Baxter Lie algebras, and also show that they give rise to solutions of the classical reflection equation for certain Lie bialgebras determined by the relative Rota-Baxter operators. In particular, involutive automorphisms on pre-Lie algebras and post-Lie algebras naturally lead to reflections on the induced relative Rota-Baxter Lie algebras.  Finally, we derive Poisson Lie groups and Poisson homogeneous spaces from quadratic Rota-Baxter Lie algebras and relative Rota-Baxter Lie algebras.
\end{abstract}

\renewcommand{\thefootnote}{}
\footnotetext{2020 Mathematics Subject Classification.
17B38, 
17B62, 
53D17
}

\keywords{the classical reflection equation, Lie bialgebra,  quadratic Rota-Baxter Lie algebra,  Poisson homogeneous space}

\maketitle

\tableofcontents
\section{Introduction}

The purpose of this paper is to construct solutions of the classical reflection equation and Poisson homogeneous spaces using Rota-Baxter Lie algebras.

\subsection{The classical reflection equation and Poisson homogeneous spaces}

  A Poisson Lie group is a Lie group with a Poisson structure such that the group multiplication is a Poisson
map. It was introduced by Drinfeld as the semi-classical limit of quantum groups.
A Poisson homogeneous space is a Poisson manifold with a transitive Poisson Lie group action, such that the action is a Poisson map.
They can be classified in terms of Dirac structures (Lagrangian subalgebras) of the Lie bialgebra associated with a Poisson Lie group \cite{D2}. Every Poisson homogeneous space can be integrated to a symplectic groupoid \cite{BIL}.

The quantum reflection equation was first introduced by Cherednik in the study of factorizing scattering on the half line \cite{C}, and by Sklyanin in the investigation of quantum integrable models with boundary conditions \cite{Sk}. As Drinfeld's quantum groups \cite{D1} are closely related with the  Yang-Baxter equation,  quantum symmetric pairs give universal solutions of the quantum reflection equation \cite{BK, BW}.
 To construct integrable Hamiltonian systems on Poisson homogeneous spaces, Schrader  introduced the classical reflection equation in \cite{S}, which is the semiclassical limit of Sklyanin's quantum reflection equation after some appropriate choices. To be specific,  for a coboundary Poisson Lie group $(G, r)$ with $r\in \g\otimes \g$ (equivalently a coboundary Lie bialgebra $(\g,r)$), a
 Lie algebra automorphism $\tau: \g\to \g$ is called a solution of the classical reflection equation  if it satisfies
 \[(\tau\otimes \tau) (r)+r-(\tau\otimes \iii_\g+\iii_\g\otimes \tau)(r)=0.\]

 The fixed points set $\h$ of a solution $\tau$ of the classical reflection equation is a coideal subalgebra of the Lie bialgebra $(\g,r)$. Denote by $H$ the  Lie group integrating the Lie algebra  $\h$. Then the homogeneous space $G/H$ inherits a Poisson structure from $G$ so that it is a Poisson homogeneous space.  Taking particular $G$ and $\tau$, Schrader recovered the semi-classical limit of the quantum integrable system, the XXZ spin chain with reflecting boundary conditions.

 Recently the reflection equation was extensively studied. Caudrelier and Zhang introduced the set-theoretical reflection equation in \cite{CZ} motivated from the study of integrable systems.  Combinatorial solutions to the reflection equation were studied by Smoktunowicz,  Vendramin and
 Weston in \cite{SVW}. Doikou and  Smoktunowicz explored the relation between the set-theoretical  reflection equation and quantum group symmetries \cite{DS}. Lebed and  Vendramin used the reflection equation as a tool for studying set-theoretical solutions to the Yang-Baxter equation \cite{LV}. In \cite{AMS}, the authors studied reflections to set-theoretic solutions of the Yang-Baxter equation by exploring their connections with their derived solutions. 

 \subsection{Rota-Baxter operators and solutions of the classical Yang-Baxter equation}

The Rota-Baxter operator on associative algebras generalizes the integral operator on continuous functions, which
was introduced in 1960 by Baxter in his probability study of fluctuation theory \cite{Bax}.  Its connection with combinatorics was
then found in the 1960s and 1970s by Cartier and Rota's school.  In the area of mathematical physics,  Rota-Baxter algebras  appeared in Connes and Kreimer's work  in the late 1990s in their Hopf algebra approach to renormalization of quantum field theory \cite{CK}. Independently, the Rota-Baxter operator on Lie algebras naturally appeared in Belavin-Drinfeld and Semenov-Tian-Shansky's works \cite{BD, STS}.  A Rota-Baxter operator on a Lie algebra $\g$ is a linear map $B: \g\to \g$ satisfying
\[[B(x),B(y)]_\g=B([B(x),y]_\g+[x,B(y)]_\g+\lambda [x,y]_\g),\qquad \forall x,y\in \g,\]
where $\lambda$ is a scalar and called the weight of the Rota-Baxter operator. To better understand the classical Yang-Baxter equation and
 related integrable systems, the more general notion of an \oop (later also called
a relative Rota-Baxter operator)
on a Lie algebra $\g$ with respect to a representation $\rho:\g\to\gl(V)$ was introduced by Kupershmidt~\cite{Ku}.

There is a close relationship between (relative) Rota-Baxter Lie algebras and solutions of the classical Yang-Baxter equation, and we summarize it as follows:
\begin{itemize}
  \item[(i)] A quadratic Rota-Baxter Lie algebra of weight $0$ yields a skew-symmetric  solution of the classical Yang-Baxter equation on the underlying Lie algebra \cite{STS,G1};

   \item[(ii)]In \cite{Bai}, Bai showed that the skew-symmetrization of a relative Rota-Baxter operator $T:V\to \g$ of weight 0 gives a skew-symmetric  solution of the classical Yang-Baxter equation in the semidirect product Lie algebra $\g\ltimes_{\rho^*} V^*$. The nonzero weight case was studied in \cite{BNG};

    \item[(iii)] A quadratic Rota-Baxter Lie algebra $(\g, B, S)$ of nonzero weight gives rise to a factorizable Lie bialgebra $(\g,r_B)$ \cite{LS}. See Theorem \ref{converse} for details. See also \cite{G2,Gon25} for similar studies.
\end{itemize}

Rota-Baxter operators also have various applications in the areas of noncommutative
symmetric functions and noncommutative Bohnenblust-Spitzer identities \cite{Fard,Yu-Guo}, splitting of operads \cite{BBGN}, double Lie algebras \cite{GG,GK} and etc. In particular, pre-Lie algebras and post-Lie algebras (\cite{Val}) are the underlying structures of (relative) Rota-Baxter operators \cite{Bai,BGN2010}. See the book \cite{Gub} for more details.


 \subsection{Main results and outline of the paper}

This paper aims to deepen and advance the applications of Rota-Baxter operators in mathematical physics. We establish the connection between Rota-Baxter operators and solutions of the classical reflection equation,
building upon the known relation  with the classical Yang-Baxter equation. Moreover, we  construct explicit examples of Poisson Lie groups and Poisson homogeneous spaces from Rota-Baxter operators.

Explicitly, we first introduce the notion of a  reflection on a quadratic Rota-Baxter Lie algebra $(\g, B, S)$, which is a Lie algebra automorphism $\tau: \g\to \g$ satisfying some compatibility conditions. They yield skew-symmetric solutions of the classical reflection equation for the corresponding Lie bialgebra $(\g,r_B)$, coideal subalgebras of $(\g,r_B)$ and  Rota-Baxter subalgebras of the descendent Rota-Baxter Lie algebra $(\g_B,B)$.  Reflections on the double of a quadratic Rota-Baxter Lie algebra are studied in detail, which include $\mathrm{sl}(n,\mathbb{C})$ with reflections given by $X\mapsto -\overline{X}^T$ and $X\mapsto -X^T$ as examples. As applications, reflections on an arbitrary Rota-Baxter Lie algebra $(\g,B)$ of weight $\lambda$ are naturally defined so that they give rise to reflections on the quadratic Rota-Baxter Lie algebra $(\g\ltimes_{\ad^*} \g^*, B\oplus (-\lambda \iii_{\g^*}-B^*), \huaS)$. Then similar ideas have been applied to obtain reflections on relative Rota-Baxter Lie algebras with both zero weight and nonzero weight. In particular, involutive pre-Lie algebra and post-Lie algebra automorphisms give reflections on the corresponding relative Rota-Baxter Lie algebras.  We also derive the Lie bialgebra structures arising from relative Rota-Baxter Lie algebras of arbitrary weight. In the end, we construct explicit Poisson Lie groups and Poisson homogeneous spaces from (relative) Rota-Baxter Lie algebras and reflections on them.

This paper is organized as follows. In Section \ref{sec:base}, we recall  Lie bialgebras, the classical reflection equations, Poisson Lie groups and Poisson homogeneous spaces. In Section \ref{sec:qua}, we introduce the notion of reflections on quadratic Rota-Baxter Lie algebras and show that reflections lead to solutions of the classical reflection equation,  coideal subalgebras and Rota-Baxter Lie subalgebras. Then reflections on Rota-Baxter Lie algebras are introduced so that they induce reflections on the semi-direct product quadratic Rota-Baxter Lie algebra. Section \ref{sec:rel} deals with reflections on relative Rota-Baxter Lie algebras of weight zero and nonzero weight, respectively.  In Section \ref{sec:PL}, Poisson Lie groups related with Rota-Baxter Lie algebras and Poisson homogeneous spaces generated from reflections on Rota-Baxter Lie algebras are expressed explicitly.

\vspace{3mm}

\noindent
{\bf Acknowledgements. } This research is supported by NSFC (12471060, W2412041). We give warmest thanks to Maxim Goncharov and Xiaomeng Xu for helpful discussions.

\section{The classical reflection equation}\label{sec:base}

In this section, we briefly recall the classical reflection equation in a coboundary Lie bialgebra introduced by Schrader in \cite{S} and Poisson homogeneous spaces.

A {\bf Lie bialgebra} is a Lie algebra $\g$ with a linear map $\delta: \g\to \wedge^2 \g$ such that the dual map $\delta^*: \wedge^2 \g^*\to \g^*$ defines a Lie bracket on $\g^*$ and
\[\delta[x,y]_\g=[\delta (x),y]_\g+[x, \delta(y)]_\g,\qquad \forall x,y\in \g.\]
Here the Lie bracket on $\g$ is extended to $\oplus_k(\wedge^k \g)$ by the Leibniz rule. A Lie bialgebra is denoted by $(\g,\g^*)$ or $(\g,\delta)$.

Let $\g$ be a Lie algebra.  For $r=\sum_i a_i\otimes b_i\in \g\otimes \g$, we introduce the notations
$r^{21}=\sum_i b_i\otimes a_i$, 
\[r_{12}=\sum_{i} a_i\otimes b_i\otimes 1,\qquad r_{13}=\sum_{i} a_i\otimes 1\otimes b_i,\qquad r_{23}=\sum_{i} 1\otimes a_i\otimes b_i,\]
and 
\[[r_{12}, r_{13}]=\sum_{i,j} [a_i,a_j]_\g\otimes b_i\otimes b_j, \quad [r_{13}, r_{23}]=\sum_{i,j} a_i\otimes a_j\otimes [b_i,b_j]_\g,\quad [r_{12},r_{23}]=\sum_{i,j} a_i\otimes[b_i, a_j]_\g\otimes b_j.\]

For a Lie algebra $\g$, given an element $r\in \g\otimes \g$, define $\delta: \g\to \wedge^2 \g$ by $\delta(x)=[x,r]$. Then $\delta^*$ defines a Lie algebra structure on $\g^*$ if and only if the symmetric part $r+r^{21}\in \g\otimes \g$ and
\[[r,r]:=[r_{12},r_{13}]+[r_{13},r_{23}]+[r_{12},r_{23}]\in \g\otimes \g\otimes \g\]
are $\ad$-invariant, respectively. The pair $(\g,\g^*)$ or $(\g,\delta)$ composes a Lie bialgebra, called a {\bf coboundary Lie bialgebra} and denoted by $(\g, \g_r^*)$,  or $(\g,r)$.

In particular, if $[r,r]=0$, we say that $r$ is a solution of the {\bf classical Yang-Baxter equation} and the Lie bialgebra $(\g,r)$ is called {\bf quasitriangular}.  Introduce $r_+: \g^*\to \g$ by $r_+(\xi)=r(\xi,\cdot)$ for $\xi\in \g^*$ and $r_-:=-r_+^*$. Denote by $I$ the linear map
\[I=r_+-r_-: \g^*\to \g.\]
Note that $I^*=I$ and it is the operator form of  $r+r^{21}$. So the $\ad$-invariance  is equivalent to
\begin{equation*}\label{eq:invI}
  I\circ \ad_x^*=\ad_x\circ I,\quad \forall x\in \g.
\end{equation*}
Quasitriangular Lie bialgebras have two important particular cases:
\begin{itemize}
\item[\rm (1)] If $r$ is skew-symmetric ($I=0$), then  $(\g,r)$ called a {\bf triangular} Lie bialgebra;
\item[\rm (2)] If  $I$ is an isomorphism of vector spaces, then $(\g,r)$ is called a {\bf factorizable Lie bialgebra}; see \cite{RS}.
  \end{itemize}

Given a Lie bialgebra $(\g,\delta)$, a Lie subalgebra $\h\subset \g$ is called a {\bf coideal} if \[\delta(\h)\subset \g\otimes \h\oplus \h\otimes \g,\]
which is equivalent to that $\h^\perp$, the annihilator space of $\h$,  is a Lie subalgebra of $\g^*$. Such an $\h$ is also called a {\bf coisotropic subalgebra}; see \cite{Lu}.

One  main theorem in \cite{S} is as follows:
\begin{thm}\label{Theorem 1}\rm (\cite{S})
Let $(\g,r)$ be a coboundary Lie bialgebra and $\tau: \g\to \g$ a Lie algebra automorphism. The fixed point set
\[\h=\g^\tau=\{x\in \g|\tau(x)=x\}\]
is a coideal subalgebra of $(\g,r)$ if and only if
\[C_\tau(r):=(\tau\otimes \tau)(r)+r-(\tau\otimes \textbf{1}_\g+\textbf{1}_\g\otimes \tau)(r)\]
is $\h$-invariant in $\g\otimes \g$.
\end{thm}
\begin{defi}
Let $(\g,r)$ be a coboundary Lie bialgebra and $\tau: \g\to \g$ a Lie algebra automorphism.
The equation \begin{eqnarray}\label{re}
(\tau\otimes \tau)(r)+r-(\tau\otimes \textbf{1}_\g+\textbf{1}_\g\otimes \tau)(r)=0
\end{eqnarray}
is called the {\bf classical reflection equation} and $\tau$ is called a {\bf solution} of the classical reflection equation. \end{defi}

\begin{rmk}\label{relaxCRE}
In terms of $r_+: \g^*\to \g$ defined by $r_+(\xi)=r(\xi,\cdot)$,  the equation \eqref{re} is equivalent to
\begin{eqnarray}\label{reeq}
(\tau-\textbf{1}_\g)\circ r_+\circ (\tau^*-\textbf{1}_{\g^*})=0.
\end{eqnarray}
And $\h=\g^\tau$ is a coideal subalgebra of $(\g, r)$ if and only of
\[\ad_x\circ (\tau-\textbf{1}_\g)\circ r_+\circ (\tau^*-\textbf{1}_{\g^*})-(\tau-\textbf{1}_\g)\circ r_+\circ (\tau^*-\textbf{1}_{\g^*})\circ \ad_x^*=0, \quad \forall x\in \h.\]
\end{rmk}

Poisson Lie groups are the semi-classical limits of quantum groups \cite{D1}.
A {\bf Poisson Lie group} is a Lie group $G$ with a Poisson bivector field $\pi\in \mathfrak{X}^2(G)$  ($[\pi,\pi]=0$) such that the multiplication $G\times G\to G$ is a Poisson map, where $G\times G$ is equipped with the product Poisson structure. It is by Drinfeld that there is a one-one correspondence between connected and simply-connected Poisson Lie groups and Lie bialgebras.

Let $(G,\pi)$ be a Poisson Lie group and $(P,\pi_P)$ a Poisson manifold. A {\bf Poisson action} of $G$ on $P$ is an action $\rho: G\times P\to P$ such that $\rho$ is a Poisson map, where $G\times P$ is equipped with the product Poisson structure. If the action is transitive, $(P,\pi_P)$ is called a {\bf Poisson homogeneous space} of $(G,\pi)$.

For a closed subgroup $H$ of a Poisson Lie group $(G,\pi)$, denote its Lie algebra by $\h$. Then
\begin{itemize}
\item[\rm (i)] $H$ is a Poisson subgroup of $G$ if and only if the annihilator space $\h^\perp \subset \g^*$ is an ideal;
\item[\rm (ii)] if $\h^\perp\subset \g^*$ is a Lie subalgebra, i.e.  $\h$ is a coideal subalgebra of $(\g, \g^*)$, then there is a unique Poisson structure $\pi_{G/H}$ on $G/H$ such that the projection $G\to G/H$ is a Poisson map. We thus obtain a Poisson homogeneous space $(G/H, \pi_{G/H})$ of $(G,\pi)$; see \cite{D2}.
\end{itemize}

In this paper, we only focus on this simple class of Poisson homogeneous spaces. We refer to \cite{D2} for the classification of Poisson homogeneous spaces.

\section{Reflections on quadratic Rota-Baxter Lie algebras and solutions of the classical reflection equation}\label{sec:qua}

In this section, we study reflections on quadratic Rota-Baxter Lie algebras, by which we construct solutions of the classical reflection equation in the corresponding triangular (weight 0 case) and factorizable Lie bialgebra (nonzero weight case).

\subsection{Reflections on quadratic Rota-Baxter Lie algebras}\label{qua}

A linear map $B:\g\lon\g$ is called a {\bf Rota-Baxter operator of weight $\lambda$} on a Lie algebra $\g$ if
\begin{equation*}
[B(x), B(y)]_\g=B([B(x),y]_\g+[x,B(y)]_\g+\lambda[x,y]_\g), \quad \forall x,y\in\g.
\end{equation*}
Let $(\g, [\cdot,\cdot]_\g)$ be a Lie algebra  and $B:\g\lon\g$   a   Rota-Baxter operator of weight $\lambda$ on   $\g$. Then there is a new Lie bracket $[\cdot,\cdot]_B$ on $\g$ defined by
$$
[x,y]_B=[B(x),y]_\g+[x,B(y)]_\g+\lambda[x,y]_\g.
$$
The Lie algebra $(\g,[\cdot,\cdot]_B)$ is called the {\bf descendent Lie algebra}, and denoted by $\g_B$. It is obvious that $B$ is a Lie algebra homomorphism from $\g_B$ to $\g$:
$$
B[x,y]_B=[B(x),B(y)]_\g.
$$
Recall that a nondegenerate symmetric bilinear form $S\in \otimes ^2\g^*$ on a Lie algebra $\g$ is said to be invariant if
\begin{eqnarray}
\label{RBmanin1}S([x,y]_\g,z)+S(y,[x,z]_\g)&=&0,\quad \forall x,y,z \in \g.
\end{eqnarray}
A quadratic Lie algebra $(\g,S)$ is a Lie algebra $\g$ equipped with a nondegenerate symmetric invariant bilinear form $S\in \otimes ^2\g^*$.
\begin{defi}\label{defi:qua}\cite[Definition 2.4]{LS} Let $(\g,[\cdot,\cdot]_\g,B)$ be a  Rota-Baxter Lie algebra of weight $\lambda$, and $S\in \otimes ^2\g^*$ a nondegenerate symmetric bilinear form.   The triple  $(\g,B,S)$ is called a {\bf quadratic Rota-Baxter Lie algebra of weight $\lambda$} if $(\g,S)$ is a quadratic Lie algebra and the following compatibility condition holds:
\begin{eqnarray}
\label{RBmanin}
S( x, {B} y)+S({B}x, y)+\lambda S(x,y)&=&0,\quad \forall x,y\in \g.
\end{eqnarray}
\end{defi}

There is a close connection between quadratic Rota-Baxter Lie algebras and triangular/factorizable Lie bialgebras; see \cite{G2} and \cite[Theorem 2.10]{LS} for more details.

\begin{thm}\label{converse}
Let $(\g, B,S)$ be a quadratic Rota-Baxter Lie algebra of weight $\lambda$, and $I_S:\g^*\to\g$ the induced linear isomorphism given by $\langle I_S^{-1} x,y\rangle:=S(x,y)$.

\begin{itemize}
  \item[{\rm(i)}] If $\lambda=0,$ then $r_B\in \wedge^2 \g$ defined by
\[r_{B+}=B\circ I_S,\qquad r_{B+}(\xi)=r_B(\xi,\cdot),\quad \forall\xi\in \g^*\]
is a skew-symmetric solution of  the classical Yang-Baxter equation and it gives rise to a triangular Lie bialgebra $(\g, r_B)$.

 \item[{\rm(ii)}] If $\lambda\neq 0,$ then  $r_B\in \g\otimes \g$ determined by
\begin{eqnarray}\label{rfromB}
r_{B+}:=\frac{1}{\lambda}(B+\lambda \iii_\g )\circ I_S:\g^*\to \g, \quad r_{B+}(\xi)=r_B(\xi,\cdot),\quad \forall\xi\in \g^*
\end{eqnarray}
satisfies the classical Yang-Baxter equation, and gives rise to a  quasitriangular Lie bialgebra $(\g,r_B)$, which is factorizable.
\end{itemize}
\end{thm}

\begin{defi}\label{refqua}
A {\bf reflection} on a quadratic Rota-Baxter Lie algebra $(\g, B,S)$ is a Lie algebra automorphism $\tau: \g\to \g$ satisfying
\begin{eqnarray}
\label{ref1}
\tau\circ B\circ \tau-B+\tau\circ B-B\circ \tau+\lambda \tau^2-\lambda \textbf{1}_\g&=&0,\\
\label{tauinv}S(\tau x,y)+S(x,\tau y)&=&0.
\end{eqnarray}
\end{defi}
\begin{rmk}
Equations \eqref{ref1}  and \eqref{tauinv} can be reformulated as
\begin{eqnarray}\label{ref1eq}
(\tau-\textbf{1}_\g)\circ (B+\lambda\textbf{1}_\g)\circ (\tau+\textbf{1}_\g)&=&0,\\
\label{tauinveq}I_S \circ \tau^*+\tau\circ I_S&=&0,
\end{eqnarray}
where $\textbf{1}_\g:\g\to \g$ is the identity map and $I_S:\g^*\to \g$ is the isomorphism induced by $S$.
\end{rmk}

\begin{rmk}\label{relaxad}
  In the current definition of a reflection on a quadratic Rota-Baxter Lie algebra, $\tau$ is required to be skew-symmetric with respect to the bilinear form $S$. In fact, there is an alternative definition in which $\tau$ is symmetric. More precisely, one can define a reflection on a quadratic Rota-Baxter Lie algebra to be a Lie algebra automorphism $\tau:\g\to\g$ such that
  \begin{eqnarray*}
(\tau-\textbf{1}_\g)\circ (B+\lambda\textbf{1}_\g)\circ (\tau-\textbf{1}_\g)&=&0,\\
I_S \circ \tau^*-\tau\circ I_S&=&0,
\end{eqnarray*}
and all the results in the sequel also hold for this definition. On the other hand,   \eqref{ref1eq} can be relaxed to the $\ad$-invariance condition: $[\ad_x, (\tau-\textbf{1}_\g)\circ (B+\lambda\textbf{1}_\g)\circ (\tau+\textbf{1}_\g)]=0$ for $x\in \ker (\tau-\iii_\g)$ for our later application. We will give Example \ref{invariant} to illustrate this fact.
\end{rmk}

\begin{ex}\label{rmkinvolution}
Given an involution $\tau: \g\to \g$, i.e. $\tau^2=\iii_\g$, we have the decomposition $\g=\h\oplus \frkp$ of $\h$-modules, where $\h$ and $\frkp$ are the $\pm 1$-eigenspaces of $\tau$. If we decompose the Rota-Baxter operator $B$ into the form
\[B=\begin{pmatrix}
B_{\h\h} & B_{\frkp\h}\\
B_{\h\frkp} &B_{\frkp \frkp}
\end{pmatrix}: \h\oplus \frkp \to \h\oplus \frkp,\]
where $B_{\h\h}: \h\to \h,B_{\h\frkp}:\h\to \frkp,B_{\frkp\h}: \frkp\to \h$ and $B_{\frkp\frkp}: \frkp\to \frkp $ are linear maps.
Then \eqref{ref1} holds if and only if $B_{\h \frkp}=0$ and \eqref{tauinv} holds if and only if $\h$ and $\frkp$ are isotropic with respect to $S$.

So an involutive Lie algebra automorphism $\tau: \g\to \g$ is a reflection on a quadratic Rota-Baxter Lie algebra $(\g,B,S)$ if and only if $B_{\h \frkp}=0$ and $\h,\frkp$ are isotropic with respect to $S$.

 If moreover, the involutive Lie algebra automorphism $\tau$ satisfies $\tau\circ B\circ \tau=B$, namely, $B_{\h\frkp}=B_{\frkp \h}=0$, then $(\g, r_B,-\tau)$ is a symmetric Lie bialgebra studied in \cite{Xu}.
\end{ex}

\begin{rmk}
For a quadratic Rota-Baxter Lie algebra $(\g,B,S)$, there associates a Lie bialgebra $(\g,r_B)$ and furthermore a Manin triple $(\g\bowtie \g_{r_B}^*, \g,\g_{r_B}^*)$. Then an involutive reflection $\tau$ on  $(\g,B,S)$ gives rise to an anti-invariant Manin triple twist $\tau\oplus (-\tau^*)$ for  $(\g\bowtie \g_{r_B}^*, \g,\g_{r_B}^*)$, where the notion of Manin triple twists was introduced in \cite{BC} to study coideal subalgebras of a Hopf algebra.
\end{rmk}

 The first important role that reflections on quadratic Rota-Baxter Lie algebras play is that they give rise to solutions of the classical reflection equation.
\begin{thm}\label{quathm}
Let $\tau$ be a reflection on a quadratic Rota-Baxter Lie algebra $(\g,B,S)$ of weight $\lambda$.
\begin{itemize}
 \item[{\rm(i)}]
In the case $\lambda=0$,  $\tau$ is  a skew-symmetric solution (with respect to the bilinear form $S$) of the classical reflection equation for the triangular Lie bialgebra $(\g,r_B)$ given in  (i) in Theorem \ref{converse}.
 \item[{\rm(ii)}]
 In the case $\lambda\neq0$, $\tau$ is  a skew-symmetric solution (with respect to the bilinear form $S$) of the classical reflection equation for the factorizable  Lie bialgebra $(\g,r_B)$ given in  (ii) in Theorem \ref{converse}.
\end{itemize}
\end{thm}
\begin{proof}

(i) By \eqref{ref1eq} with $\lambda=0$, we have $(\tau-\textbf{1}_\g)\circ B\circ (\tau+\textbf{1}_\g)=0$.  Then applying \eqref{tauinveq} and (i) in Theorem \ref{converse}, we find
\begin{eqnarray*}
(\tau-\textbf{1}_\g)\circ r_{B+}\circ  (\tau^*-\textbf{1}_\g)&=&(\tau-\textbf{1}_\g)\circ B\circ I_S \circ (\tau^*-\textbf{1}_\g)\\ &=&-(\tau-\textbf{1}_\g)\circ B\circ (\tau+\textbf{1}_\g)\circ I_S\\ &=&0.
\end{eqnarray*}
Namely, $\tau$ is a skew-symmetric solution of the classical reflection equation for the triangular Lie bialgebra $(\g,r_B)$.

(ii) Similarly, by \eqref{ref1eq}, \eqref{tauinveq} and (ii) in Theorem \ref{converse}, we have
\begin{eqnarray*}
(\tau-\textbf{1}_\g)\circ r_{B+} \circ (\tau^*-\textbf{1}_\g)&=&\frac{1}{\lambda}(\tau-\textbf{1}_\g)\circ (B+\lambda \textbf{1}_\g)\circ I_S \circ (\tau^*-\textbf{1}_\g)\\ &=&-\frac{1}{\lambda}(\tau-\textbf{1}_\g)\circ (B+\lambda \textbf{1}_\g)\circ (\tau+\textbf{1}_\g)\circ I_S\\
&=&0,
\end{eqnarray*}
which implies that $\tau$ is a skew-symmetric solution of the classical reflection equation for the factorizable  Lie bialgebra $(\g,r_B)$.
\end{proof}

\begin{cor}\label{Sss}
Let $\tau$ be a reflection on a quadratic Rota-Baxter Lie algebra $(\g,B,S)$ of weight $\lambda$. Denote by $\h$ the fixed point set of $\tau$. Then $\h$ is a coideal subalgebra of $(\g,r_B)$.
\end{cor}
\begin{proof}
Following from Theorem \ref{quathm} and Theorem \ref{Theorem 1}, we see that $\h^\perp\subset \g^*_{r_B}$ is a Lie subalgebra so that $\h$ is a coideal subalgebra.
\end{proof}

Another important property for reflections on quadratic Rota-Baxter Lie algebras is that they give rise to Rota-Baxter subalgebras.
Recall from \cite[Theorem 3.3]{LS} that given a Rota-Baxter Lie algebra $(\g,B)$ of weight $\lambda$, we have a Lie algebra $\g_B\bowtie \g$ such that $\g_B$ and $\g$ are Lie subalgebras and the bracket between $\g_B$ and $\g$ is \[[(\xi,0),(0,a)]_{\bowtie}=([\xi,a]_\g,[B\xi,a]_\g-B[\xi,a]_\g),\qquad \forall \xi\in \g_B,a\in \g.\]

\begin{pro}\label{h^0}
Let $\tau$ be a reflection on a quadratic Rota-Baxter Lie algebra $(\g,B,S)$ of arbitrary weight.  Denote by  $\h$ the fixed point set of $\tau$ and $\h^0:=\{x\in \g|S(x,u)=0, \forall u\in \h\}$ the orthogonal complement of $\h$ in $\g$.
Then
\begin{itemize}
\item[\rm (i)] $(\h^0, B)$ is a Rota-Baxter Lie subalgebra of $(\g_B,B)$;
\item[\rm (ii)] $\h^0\bowtie \h$ is a Lie subalgebra of $\g_B\bowtie \g$.
\end{itemize}
\end{pro}
\begin{proof}
For (i), first we show that $B(\h^0)\subset \h^0$. Indeed, for $u\in \h$, by \eqref{ref1eq}, we have
\[0=(\tau-\textbf{1}_\g)(B+\lambda \textbf{1}_\g)(\tau+\textbf{1}_\g)(u)=2(\tau-\textbf{1}_\g)(B+\lambda \textbf{1}_\g)(u),\]
which implies that $Bu+\lambda u\in \h$. Therefore, by \eqref{RBmanin}, for all $x\in \h^0$ and $u\in \h$, we have
$S(Bx,u)=-S(x,\lambda u+Bu)=0$ and thus $Bx\in \h^0$.

Then we check that $\h^0$  is a Lie subalgebra of $\g_B$. It can be deduced from Corollary \ref{Sss} and the isomorphism between $\g_{r_B}^*$ and $\g_B$. But  we prefer to provide a straightforward verification.
By the invariance of $S$ and $\h$ being a Lie subalgebra of $\g$, we have that
$[\h^0,\h]_\g\subset \h^0$. Then we claim  $\h^0=\mathrm{Im}(\tau+\textbf{1}_\g)$. Indeed, 
 since $I_S\circ \tau^*=-\tau\circ I_S$, we see
\[\h^0=I_S (\h^\perp)=I_S(\mathrm{ker}(\tau-\textbf{1}_\g)^\perp)=I_S (\mathrm{Im}(\tau^*-\textbf{1}_{\g^*}))=\mathrm{Im}(\tau+\textbf{1}_\g)\circ I_S,\]
which implies that $\h^0=\mathrm{Im}(\tau+\textbf{1}_\g)$. Together with \eqref{ref1eq}, we obtain
\begin{eqnarray}\label{ImB}
(B+\lambda \textbf{1}_\g)(\h^0)\subset \h.
\end{eqnarray}
Then
for $x,y\in \h^0$ and $u\in \h$, by the invariance of $S$ and the compatibility of $S$ and $B$, we have
\begin{eqnarray*}
S([x,y]_B,u)&=&S([Bx+\lambda x, y]_\g+[x,By]_\g,u)\\ &=&-S(y,[(B+\lambda \textbf{1}_\g)x,u]_\g)-S(By,[x,u]_\g)\\ &=&
-S(y,[(B+\lambda \textbf{1}_\g)x,u]_\g)+S(y,(B+\lambda \textbf{1}_\g)[x,u]_\g)\\ &=&0,
\end{eqnarray*}
where in the last equation we used the facts that $\h$ is a Lie subalgebra, $[\h^0,\h]_\g\subset \h^0$ and \eqref{ImB}.  Hence we see that $[x,y]_B\in \h^0$. All together, we proved that $(\h^0,[\cdot,\cdot]_B,B)$ is a Rota-Baxter Lie subalgebra of $(\g_B, B)$.

For (ii), note that $[\h^0,\h]_\g\subset \h^0$. It suffices to show that $[Bx,u]_\g-B[x,u]_\g\in \h$ for $x\in \h^0, u\in \h$. In fact, by the invariance of $S$ and the compatibility condition of $B$ and $S$, for all $y\in \h^0$, we have
\begin{eqnarray*}
S([Bx,u]_\g-B[x,u]_\g,y)&=&-S(u, [Bx,y]_\g)+S([x,u]_\g,By+\lambda y)
\\ &=&-S(u,[Bx,y]_\g+[x, By+\lambda y]_\g)=-S(u, [x,y]_B)\\
&=&0,
\end{eqnarray*}
which implies that $[Bx,u]_\g-B[x,u]_\g\in \h$ since $[x,y]_B\in \h^0$. So we see that $\h^0\bowtie \h$ is a Lie subalgebra of $\g_B\bowtie \g$. \end{proof}


\begin{pro}\label{ex}
Let $(\g,B,S)$ be an arbitrary quadratic Rota-Baxter Lie algebra of nonzero weight $\lambda$. Then
\begin{itemize}
\item [\rm (i)] the triple $(\g\oplus \g, \hat{B}, \hat{S})$ of the direct product Lie algebra $\g\oplus \g$ with $\hat{B}$ and $\hat{S}$  forms a quadratic Rota-Baxter Lie algebra of  weight $\lambda$, where
\[\hat{B}(u,x)=\big(Bu-(B+\lambda \iii_\g)x,Bu-(B+\lambda \iii_\g)x\big),\quad \hat{S}((u,x),(v,y))=S(u,v)-S(x,y).\]
And  the map $\tau: \g\oplus \g\to \g\oplus \g$ given by \[\tau(u,x)=(x,u)\] is a reflection on $(\g\oplus \g, \hat{B}, \hat{S})$.  Moreover, the Lie algebra $\h$ of fixed points of $\tau$ is $\g_{diag}:=\{(x,x)\in \g\oplus \g\}$ and its orthogonal component $\h^0=\h$ is an ideal of the descendent Lie algebra $(\g\oplus \g)_{\hat{B}}$.
\item[\rm (ii)] the triple $(\g_B\bowtie \g, B_{\bowtie}, S_{\bowtie})$ is a quadratic Rota-Baxter Lie algebra of weight $\lambda$, where 
$B_{\bowtie}, S_{\bowtie}$ are given by
\[B_{\bowtie}(\xi,x)=(0,-\lambda x),\qquad S_{\bowtie}((\xi,x),(\eta,y))=\lambda(S(\xi,y)+S(x,\eta)).\]
Moreover, the map $\tau_{\bowtie}: \g_B\bowtie \g\to \g_B\bowtie \g$ given by \[\tau_{\bowtie} (\xi,x)=(-\xi,(2B+\lambda \iii_\g)\xi+x)\]
is a reflection on $(\g_B\bowtie \g, B_{\bowtie}, S_{\bowtie})$.
\item[\rm (iii)] the double Lie algebra $(\g^*_{r_B}\bowtie \g,B_D,\huaS)$ is a quadratic Rota-Baxter Lie algebra of weight $\lambda$, where
\[B_D(\alpha,x)=(0,-\lambda x),\qquad \huaS((\alpha,x),(\beta,y))=\langle \alpha,y\rangle+\langle x,\beta\rangle.\]
The map $\tau_D: \g^*_{r_B}\bowtie \g\to\g^*_{r_B}\bowtie \g$ defined by
\[\tau_D(\alpha,x)=(-\alpha, \frac{1}{\lambda}(2B+\lambda \iii_\g)I_S \alpha+x),\]
is a reflection on $(\g^*_{r_B}\bowtie \g,B_D,\huaS)$.
\end{itemize}
\end{pro}
\begin{proof}
For (i),  it is known that  $\hat{B}$ is a Rota-Baxter operator on $\g\oplus \g$ if and only if
\[C=-\lambda \iii_{\g\oplus \g}-\hat{B}: \g\oplus \g\to \g\oplus \g,\qquad (u,x)\mapsto \big((B+\lambda \iii_\g)(x-u),B(x-u)\big)\] is a Rota-Baxter operator. Actually, since $B$ is a Rota-Baxter operator on $\g$, we have
\begin{eqnarray*}
[C(u,x),C(v,y)]_{\g\oplus \g}&=&\big([(B+\lambda \iii_\g)(x-u),(B+\lambda \iii_\g)(y-v)]_\g,[B(x-u),B(y-v)]_\g\big)\\ &=&
((B+\lambda\iii_\g)f, Bf), \end{eqnarray*}
where $f=[B(x-u), y-v]_\g+[x-u, (B+\lambda\iii_\g)(y-v)]_\g$.
And
\begin{eqnarray*}
&&C([C(u,x), (v,y)]_{\g\oplus \g}+[(u,x), (C+\lambda \iii_{\g\oplus \g})(v,y)]_{\g\oplus \g})\\ &=&C\big(([(B+\lambda \iii_\g)(x-u),v]_\g,[B(x-u),y]_\g)+([u, (B+\lambda \iii_\g)(y-v)+\lambda v]_\g,[x, B(y-v)+\lambda y]_\g)\big)\\ &=&\big((B+\lambda \iii_\g)f', Bf'),
\end{eqnarray*}
where
\[f'=[B(x-u),y-v]_\g-\lambda [x-u, v]_\g+[x-u, (B+\lambda \iii_\g)(y-v)]_\g+\lambda[x,v]_\g-\lambda [u,v]_\g=f.\]
Therefore, we showed that $C$ is a Rota-Baxter operator, so $\hat{B}$ is also a Rota-Baxter operator.  Then depending on  the invariance of $S$, we have
\begin{eqnarray*}
&&\hat{S}([(u,x),(v,y)]_{\g\oplus \g}, (w,z))+\hat{S}((v,y),[(u,x), (w,z)]_{\g\oplus \g})\\ &=&S([u,v]_\g,w)-S([x,y]_\g,z)+S(v,[u,w]_\g)-S(y, [x,z]_\g)=0.
\end{eqnarray*}
It suffices to show \eqref{RBmanin} for $\hat{S}$ and $\hat{B}$.  Indeed, since \eqref{RBmanin} holds for $S$ and $B$, we have
\begin{eqnarray*}
&&\hat{S}(\hat{B}(u,x), (v,y))+\hat{S}((u,x), (\hat{B}+\lambda\iii)(v,y))\\ &=&S(Bu-(B+\lambda \iii_\g)x,v)-S(Bu-(B+\lambda \iii_\g)x,y)+S(u, (B+\lambda \iii_\g)(v-y))-S(x, B(v-y))\\ &=&0.
\end{eqnarray*}
So we proved that $(\g\oplus \g, \hat{B}, \hat{S})$ is a quadratic Rota-Baxter Lie algebra.

It is obvious that $\tau$ is an involutive  Lie algebra automorphism. By direct calculation, we have
\begin{eqnarray*}
(\tau-\iii_{\g\oplus \g})\hat{B}(\tau+\iii_{\g\oplus \g})(u,x)&=&(\tau-\iii_{\g\oplus \g})(-\lambda (x+u),-\lambda(x+u))=0,\\
\hat{S}(\tau(u,x),(v,y))+\hat{S}((u,x),\tau(v,y))&=&S(x,v)-S(u,y)+S(u, y)-S(x,v)=0.
\end{eqnarray*}
Hence $\tau$ is a reflection on the quadratic Rota-Baxter Lie algebra $(\g\oplus \g,\hat{B}, \hat{S})$. It is easy to see that $\h=\{(x,x); x\in \g\}$ is isotropic and $\h^0=\h$. To show that $\h^0=\h$ is not only a Lie subalgebra of $\g\oplus \g$, but also an ideal, we first note that
\begin{eqnarray*}
[(u,x),(v,y)]_{\hat{B}}&=&\big([B(u-x)-\lambda x, v]_\g,[B(u-x)-\lambda x,y]_\g\big)\\ &&+\big([u, B(v-y)-\lambda y+\lambda v]_\g, [x, B(v-y)-\lambda y+\lambda y]_\g\big)\\ &=&
\big([u,v]_B-[(B+\lambda \iii_\g)x, v]_\g-[u, (B+\lambda \iii_\g)y]_\g,-[x,y]_B+[Bu,y]_\g+[x,Bv]_\g\big).
\end{eqnarray*}
Then, we obtain
\[[(x,x),(v,y)]_{\hat{B}}=\big([x, Bv]_\g-[x,(B+\lambda \iii_\g)y]_\g, [x, Bv]_\g-[x,(B+\lambda \iii_\g)y]_\g)\in \h^0,\]
which implies that $\h^0\subset (\g\oplus \g)_{\hat{B}}$ is an ideal.

For (ii) and (iii), by \cite[Proposition 3.5]{LS}, given a quadratic Rota-Baxter Lie algebra $(\g,B,S)$ with nonzero weight $\lambda$,  we have a commutative diagram of Lie algebra isomorphisms:
\begin{eqnarray*}\label{diagram}
		\vcenter{\xymatrix{
			\g_{r_B}^*\bowtie \g \ar[d]_(0.45){(\frac{1}{\lambda}I)\oplus \iii_\g}\ar[r]^(0.45){\psi} &\g\oplus \g\\
			\g_B\bowtie \g,\ar[ur]_{\phi}&
		}}
\end{eqnarray*}
where $\psi$ and $\phi$ are given by
\[\psi(\alpha,x)=(r_{B+} \alpha+x, r_{B-} \alpha+x),\qquad
\phi(\xi, x)=(B\xi+\lambda \xi+x, B\xi+x).\]
Applying (i) and the isomorphisms $\psi$ and $\phi$, we obtain the quadratic
Rota-Baxter Lie algebra structures on $\g_B\bowtie \g$ and $\g_{r_B}^*\bowtie \g$ so that $\psi$ and $\phi$ become isomorphisms of quadratic Rota-Baxter Lie algebras.  We omit the details.
\end{proof}

\begin{ex}\label{sln}
Let $\g=\mathrm{sl}(n,\mathbb{C})$ be the quadratic Lie algebra with the scalar product $S(X,Y)=\mathrm{Im} tr(XY)$. Define
\[B:\g\to \g, \quad X=(x_{ij})\mapsto -\lambda A=-\lambda(a_{ij}),\]
where \[  a_{ii}=\frac{x_{ii}+\overline{x_{ii}}}{2}, \quad a_{ij}=x_{ij}+\overline{x_{ji}},~ \mbox{if}~ i<j, ~\mbox{and}~  a_{ij}=0,~ \mbox{if}~ i>j.\]
Then  $(\g, B, S)$ is a quadratic Rota-Baxter Lie algebra. In fact, $\g$ has the decomposition $\mathrm{sl}(n,\mathbb{C})=\mathrm{su}(n)\oplus \mathrm{sb}(n,\mathbb{C})$ of vector spaces,  where
$\mathrm{su}(n):=\{X\in  \mathrm{sl}(n,\mathbb{C})|\overline{X}^T+X=0\}$ and $\mathrm{sb}(n,\mathbb{C})$ consists of upper-triangular matrices in $\mathrm{sl}(n,\mathbb{C})$ with real diagonals. Moreover, $\mathrm{su}(n)$ and $\mathrm{sb}(n,\mathbb{C})$ are isotropic Lie subalgebras of $\g$; see \cite{LPV} for details. And the above $B$ is the projection from $\mathrm{sl}(n,\mathbb{C})$ to $\mathrm{sb}(n,\mathbb{C})$ multiplied by $-\lambda$.  
Based on this, for $X=X_1+X_2, Y=Y_1+Y_2$, with $X_1,Y_1\in \mathrm{su}(n)$ and $X_2,Y_2\in \mathrm{sb}(n,\mathbb{C})$, we have
\begin{eqnarray*}
S(BX,Y)+S(X,(B+\lambda \iii_\g)Y)&=&S(-\lambda X_2, Y_1+Y_2)+S(X_1+X_2,\lambda Y_1)\\ &=&-\lambda S(X_2,Y_1)+\lambda S(X_2, Y_1)=0.
\end{eqnarray*}
So $(\g, B, S)$ is a quadratic Rota-Baxter Lie algebra. The descendent Lie bracket $[\cdot,\cdot]_B$ on $\g_B$ is
\[[X,Y]_B=[X_1+X_2,Y_1+Y_2]_B=\lambda[X_1, Y_1]_\g-\lambda[X_2,Y_2]_\g.\]

Then by Proposition \ref{ex}, we see that $(\mathrm{sl}(n,\mathbb{C})\oplus \mathrm{sl}(n,\mathbb{C}), \hat{B},\hat{S})$ is a quadratic Rota-Baxter Lie algebra and $\tau(u,x)=(x,u)$ is a reflection on it. Here
\begin{eqnarray*}
\hat{S}((U,X),(V,Y))&=&\mathrm{Im} tr(UV)-\mathrm{Im}tr(XY),\\ \hat{B}(U,X)&=&-\lambda(C, C), \qquad C=(c_{ij}),
\end{eqnarray*}
where $$c_{ii}=\frac{u_{ii}+\overline{u_{ii}}+x_{ii}-\overline{x_{ii}}}{2},\quad c_{ij}=u_{ij}+\overline{u_{ji}}-\overline{x_{ji}} ~\mbox{ for}~ i<j, ~\mbox{and}~ c_{ij}=x_{ij}~\mbox{ for}~ i>j.$$
The fixed point set of $\tau$ is $\mathrm{sl}(n,\mathbb{C})_{diag}\cong \mathrm{sl}(n,\mathbb{C})$ and its orthogonal component $\h^0=\h$ is an ideal of $(\mathrm{sl}(n,\mathbb{C})\oplus \mathrm{sl}(n,\mathbb{C}))_{\hat{B}}$.\end{ex}
\begin{ex}\label{sln2}
Let $(\g=\mathrm{sl}(n,\mathbb{C}), B, S)$ be the quadratic  Rota-Baxter Lie algebra in Example \ref{sln}. Then the involutive Lie algebra automorphism
\[\tau: \g\to \g, \qquad X\mapsto -\overline{X}^T\]
 is a reflection on  $(\mathrm{sl}(n,\mathbb{C}), B, S)$.

 To see this, note that the $\pm 1$-eigenspaces of $\tau$ are $\h=\mathrm{su}(n)$ and $\frkp$ and we have $\g=\h\oplus \frkp$, where $\frkp$ consists of traceless Hermitian matrices.  By Example \ref{rmkinvolution}, it suffices to show $B_{\h\frkp}=0$ and $S(\h,\h)=S(\frkp,\frkp)=0$. Indeed, for $X\in \h$, i.e. $\overline{X}^T=-X$, we have $x_{ij}+\overline{x_{ji}}=0$ and thus  $B(X)=0$ by the definition of $B$, which implies $B_{\h\frkp}=0$.  By Example \ref{sln} we have $S(\h,\h)=0$. For $X,Y\in \frkp$, we have $\overline{X}^T=X$ and $\overline{Y}^T=Y$. Then we find
\[\overline{tr(XY)}=tr(\overline{X}\overline{Y})=tr(X^TY^T)=tr(YX)=tr(XY),\]
so that $S(X,Y)=\mathrm{Im} tr(XY)=0$. Therefore, $\tau$ is a reflection on $(\mathrm{sl}(n,\mathbb{C}), B, S)$. Moreover, $\h^0\cong \h$ is not just a Lie subalgebra of $\mathrm{sl}(n,\mathbb{C})_{B}$, but an ideal. \end{ex}

\emptycomment{
\begin{rmk}
There is an alternative way to introduce reflections on a quadratic Rota-Baxter Lie algebra.
An {\bf invariant reflection} \hl{symmetric or invariant ?} on a quadratic Rota-Baxter Lie algebra $(\g,B,S)$ of weight $\lambda$ is a Lie algebra automorphism $\tau: \g\to \g$ such that
\begin{eqnarray}
\label{ref1anti}
\tau\circ B\circ \tau+B-\tau\circ B-B\circ \tau+\lambda \tau^2-2\lambda \tau+\lambda \iii_\g&=&(\tau-\iii_\g)(B+\lambda \iii_\g)(\tau-\iii_\g)=0,\\
\label{tauanti}S(\tau x,y)-S(x,\tau y)&=&0.
\end{eqnarray}
We could obtain parallel results as in  Theorem \ref{quathm} and Proposition \ref{h^0}. Such an invariant reflection $\tau$ is a symmetric solution of the classical reflection equation for the Lie bialgebra $(\g, r_B)$.  Denote by $\h$ the fixed point set of $\tau$ and $\h^0:=\{x\in \g|S(x,u)=0, \forall u\in \h\}$. Then $\h^0$ is a  Lie subalgebra of $(\g_B,B)$.

In fact, to make $\h^0$ a Lie subalgebra of $(\g_B,B)$, the condition \eqref{ref1anti} can be relaxed to
\begin{eqnarray}\label{relax}
[\ad_x, (\tau-\iii_\g)(B+\lambda \iii_\g)(\tau-\iii_\g)]_C=0,\quad x\in \h,
\end{eqnarray}
where $[\cdot,\cdot]_C$ is the commutator bracket of two linear maps from $\g$ to $\g$. See Remark \ref{relexCRE}.
\end{rmk}
}

\begin{ex}\label{invariant}
Consider the complex Lie algebra $\g=\mathrm{sl}(2,\mathbb{C})$ with a basis $H=\begin{pmatrix}
1& 0 \\
0 & -1
\end{pmatrix},  X=\begin{pmatrix}
0& 1 \\
0 & 0
\end{pmatrix}, Y=\begin{pmatrix}
0& 0 \\
1 & 0
\end{pmatrix}$.
Define
\begin{eqnarray*}S(aH+bX+cY, a'H+b'X+c'Y)&=&8aa'+4(bc'+cb'),\\
B(aH+bX+cY)&=&-\frac{1}{2}\lambda aH-\lambda bX.\end{eqnarray*}
Then $(\mathrm{sl}(2,\mathbb{C}), S,B)$ is a quadratic Rota-Baxter Lie algebra of weight $\lambda$, which corresponds to the factorizable Lie bialgebra $(\g, r=\frac{1}{8}(H\otimes H+4X\otimes Y))$ in \cite{K}.

The involutive Lie algebra automorphism
\[\tau: \g\to \g, \qquad U\mapsto -U^T\] satisfies the condition $S(\tau x,y)-S(x,\tau y)=0$ and
\[(\tau-\iii_\g)(B+\lambda \iii_\g)(\tau-\iii_\g)=\lambda(\iii_\g-\tau).\]
Although it does not vanish, it actually satisfies the ad-invariance condition:
\[[\ad_x, \lambda(\iii_\g-\tau)](u)=\lambda[x,u-\tau u]-\lambda (\iii_\g-\tau)[x,u]=0,\quad \forall x\in \h, u\in \g, \]
where $\h=\mathrm{so}(n,\mathbb{C})$, the fixed point set of $\tau$, is the Lie algebra of skew-symmetric matrices.
We could still obtain the result that $\h^0$ is  a Lie subalgebra of $\g_B$. Note that $\h^0$ is composed of all traceless symmetric matrices. We show $\h^0$ is not an ideal. In fact, the descendent Lie bracket $[\cdot,\cdot]_B$ on the basis is:
\[[H,X]_B=-\lambda X,\quad [H,Y]_B=-\lambda Y,\quad [X,Y]_B=0.\]
For $aH+bX+bY\in \h^0$, we have
\[[aH+bX+bY, X]_B=-\lambda a X\notin \h^0.\]
So $\h^0$ is not an ideal.
This example can be generalized to $\mathrm{sl}(n,\mathbb{C})$.
\end{ex}

\subsection{Reflections on Rota-Baxter Lie algebras}
A Rota-Baxter Lie algebra $(\g,B)$ of weight $\lambda$ gives rise to a quadratic Rota-Baxter Lie algebra
\[\big(\g\ltimes_{\ad^*} \g^*, B\oplus (-\lambda \textbf{1}_{\g^*}-B^*), \huaS\big)\] of the same weight, where $\huaS$ is the symmetric bilinear form induced by the natural pairing:
\[\huaS(x+\xi,y+\eta)=\langle x,\eta\rangle+\langle \xi, y\rangle,\quad \forall x,y\in \g, \xi,\eta\in \g^*.\]
Note that in this special case, the isomorphism $I_\huaS: (\g\oplus \g^*)^*=\g\oplus\g^*\to \g\oplus\g^*$ is exactly the identity map. Then applying the results in Subsection \ref{qua}, we see that Rota-Baxter operators give solutions of the classical Yang-Baxter equation. We also study reflections and solutions of the classical reflection equations in this case.

\begin{pro}\label{r}
Let $(\g,B)$ be a Rota-Baxter Lie algebra of  weight $\lambda$, and $\g\ltimes_{\ad^*} \g^*$ be the semi-direct product Lie algebra with respect to the coadjoint action.

\begin{itemize}
  \item[{\rm(i)}] In the case $\lambda=0$,  $r_B\in \wedge^2(\g\ltimes_{\ad^*} \g^*)$ given by
\[r_B(x+\xi,y+\eta)= \langle Bx,\eta\rangle-\langle \xi, By\rangle\]
is a skew-symmetric solution of the classical Yang-Baxter equation in the Lie algebra $\g\ltimes_{\ad^*} \g^*$.
It further gives a triangular Lie bialgebra $(\g\ltimes_{\ad^*} \g^*, r_B)$.

 \item[{\rm(ii)}] In the case $\lambda\neq0$,  $r_B\in \otimes^2(\g\ltimes_{\ad^*} \g^*)$ given by
\[r_B(x+\xi,y+\eta)=\frac{1}{\lambda}\big(\langle Bx,\eta\rangle-\langle \xi, By\rangle+\lambda\langle x,\eta\rangle\big),\]
is a solution of the classical Yang-Baxter equation in the Lie algebra $\g\ltimes_{\ad^*} \g^*$.
It further gives a factorizable Lie bialgebra $(\g\ltimes_{\ad^*} \g^*, r_B)$.
\end{itemize} \end{pro}

\begin{proof}
By Theorem \ref{converse},  the quadratic Rota-Baxter Lie algebra $(\g\ltimes_{\ad^*} \g^*, B\oplus (-\lambda \textbf{1}_{\g^*}-B^*), \huaS)$   induces  solutions of the classical Yang-Baxter equation in the Lie algebra $\g\ltimes_{\ad^*} \g^*$.

(i) In the case of $\lambda=0$, we have
$$
r_{B+}=  B\oplus ( -B^*) :\g\oplus \g^*\to \g\oplus \g^*,
$$
which implies that
\[r_B(x+\xi,y+\eta)=\huaS(r_{B+}( x+  \xi),y+\eta)=\huaS(Bx-B^*\xi,y+\eta\rangle=\langle Bx,\eta\rangle-\langle \xi, By\rangle.\]

(ii) In the case of $\lambda\neq0$,
the $r$-matrix is given by
\begin{eqnarray}\label{rB}
r_{B+}=\frac{1}{\lambda} \big(B\oplus (-\lambda \textbf{1}_{\g^*}-B^*)+\lambda \textbf{1}_{\g\oplus \g^*}\big):\g\oplus \g^*\to \g\oplus \g^*,
\end{eqnarray}
which implies that
\[r_B(x+\xi,y+\eta)=\frac{1}{\lambda}\huaS(Bx-\lambda \xi-B^*\xi+\lambda x+\lambda \xi,y+\eta)=\frac{1}{\lambda}\big(\langle Bx,\eta\rangle-\langle \xi, By\rangle+\lambda\langle x,\eta\rangle\big).\]
The proof is finished.
\end{proof}

Denote by $(\g\ltimes_{\ad^*} \g^*)_{B\oplus (-\lambda \iii_{\g^*}-B^*)}$ the descendent Lie algebra on $\g\oplus \g^*$. By \cite{LS}, we see that
\[(\g\ltimes_{\ad^*} \g^*)_{B\oplus (-\lambda \textbf{1}_{\g^*}-B^*)}=\g_B\ltimes_{\rho_B} \g^*,\]
where $\rho_B: \g_B \to \mathrm{End}(\g^*)$ is the representation of the descendent Lie algebra $\g_B$ on $\g^*$ given by \[\rho_B(x)(\xi):=\ad_{Bx}^*\xi-\ad_{x}^*B^*\xi,\quad \forall x\in \g, \xi\in \g^*.\]
This Lie algebra is actually isomorphic to the Lie algebra $(\g\oplus \g^*)_{r_B}$ induced by the classical $r$-matrix $r_B$. We refer to \cite[Corollary 2.8]{LS} for details. Here we give a direct proof for this particular case.
\begin{lem}\label{iso}
There is a Lie algebra isomorphism from the dual Lie algebra induced by the $r$-matrix $r_B$ given in Proposition \ref{r} to the descendent Lie algebra of the Rota-Baxter Lie algebra $(\g\ltimes_{\ad^*} \g^*, B\oplus (-\lambda \iii_{\g^*}-B^*))$:
\begin{itemize}   \item[{\rm(i)}] When $\lambda=0$, we have $(\g\oplus \g^*)_{r_B}= (\g\ltimes_{\ad^*} \g^*)_{B\oplus (-B^*)}=\g_{B}\ltimes_{\rho_B}\g^*$.
    \item[{\rm(ii)}] When $\lambda\neq 0$, the isomorphism is
$$
\frac{1}{\lambda} \iii_{\g\oplus  \g^*}: (\g\oplus \g^*)_{r_B}\to (\g\ltimes_{\ad^*} \g^*)_{B\oplus (-\lambda \iii_{\g^*}-B^*)}=\g_{B}\ltimes_{\rho_B}\g^*.
$$  \end{itemize}\end{lem}
\begin{proof}
By definition, the induced Lie bracket of $r_B$ on the dual space of $\g\ltimes_{\ad^*} \g^*$ takes the formula:  \[[a,b]_{r_B}=\add_{r_{B+} (a)}^*b-\add_{r_{B-} (b)}^*a, \qquad a,b\in \g\oplus  \g^*,\]
 where $\add^*$ is the coadjoint action of $\g\ltimes_{\ad^*} \g^*$ on its dual space.  To be explicit, we have
 \[\add^*_x y=[x,y]_\g,\quad \add_x^*\xi=\ad_x^*\xi,\quad \add_\xi^*x=-\ad_x^*\xi,\quad \add_\xi^*\eta=0,\quad \forall x,y\in \g, \xi,\eta \in \g^*.\]
For (i), we have $r_{B+}(x+\xi)=r_{B-}(x+\xi)=Bx-B^*\xi$, and thus
\begin{eqnarray*}[x+\xi, u+\eta]_{r_B}&=&\add^*_{Bx-B^*\xi} (y+\eta)-\add^*_{By-B^*\eta} (x+\xi)\\
&=&[Bx,y]_\g+\ad_{Bx}^*\eta+\ad_{y}^* B^*\xi-[By,x]_\g-\ad_{By}^*\xi-\ad_x^*B^*\eta.\end{eqnarray*}
Therefore, we have $(\g\oplus \g^*)_{r_B}=\g_{B}\ltimes_{\rho_B}\g^*$.

 For (ii),
by \eqref{rB}, we have
\[r_{B+}(x+\xi)=\frac{1}{\lambda}(Bx+\lambda x-B^*\xi),\qquad r_{B-}(x+\xi)=\frac{1}{\lambda}(Bx-\lambda \xi-B^*\xi).\]
It follows that
\begin{eqnarray*}
{[x,y]_{r_B}}&=&\frac{1}{\lambda}(\add^*_{\lambda x+Bx} y-\add^*_{By} x)=\frac{1}{\lambda}([Bx,y]+[x,By]+\lambda[x,y])=\frac{1}{\lambda}[x,y]_B;\\
{[x,\xi]_{r_B}}&=&\frac{1}{\lambda}(\add^*_{\lambda x+Bx} \xi+\add^*_{\lambda \xi+B^*\xi}x)=\frac{1}{\lambda}(\ad_{\lambda x+Bx}^*\xi- \ad_x^*(\lambda \xi+B^*\xi))=\frac{1}{\lambda}(\ad_{Bx}^*\xi-\ad_x^*B^*\xi);\\
{[\xi,\eta]_{r_B}}&=&\frac{1}{\lambda}(\add^*_{-B^*\xi} \eta+\add^*_{\lambda \eta+B^*\eta} \xi)=0.
\end{eqnarray*}
This implies that $\frac{1}{\lambda}[a,b]_{r_B}=[\frac{1}{\lambda} a,\frac{1}{\lambda} b]_{B\oplus (-\lambda \textbf{1}_{\g^*}-B^*)}$ for $a,b\in \g\oplus \g^*$. So we see that $\frac{1}{\lambda}\textbf{1}_{\g\ltimes_{\ad^*} \g^*}$ is a Lie algebra isomorphism.
\end{proof}

\begin{defi}\label{reflectiononRB}
A {\bf reflection} on a Rota-Baxter Lie algebra $(\g,B)$ of arbitrary weight $\lambda$ is a Lie algebra automorphism $\tau: \g\to \g$ such that
\begin{eqnarray}
\label{homo}\tau[\tau x,y]_\g-[x,\tau y]_\g&=&0,\\
\label{involution}\lambda \tau^2-\lambda \textbf{1}_\g&=&0,\\
\label{ref}\tau \circ B\circ \tau-B+\tau\circ B-B\circ \tau&=&0.
\end{eqnarray}
\end{defi}
\begin{ex}If $\tau: \g\to \g$ is an involutive automorphism  (i.e. $\tau^2=\textbf{1}_\g$) and $B\circ \tau=\tau\circ B$, then it is a reflection on a
Rota-Baxter Lie algebra $(\g,B)$.
\end{ex}

A reflection $\tau$ on a  Rota-Baxter Lie algebra $(\g,B)$ naturally leads to a reflection on the quadratic Rota-Baxter Lie algebra $(\g\ltimes_{\ad^*} \g^*, B\oplus (-\lambda \textbf{1}_{\g^*}-B^*), \huaS)$.
\begin{thm}\label{rrr}
Let $(\g,B)$ be a Rota-Baxter Lie algebra of weight $\lambda$ and $\tau:\g\to \g$ a Lie algebra automorphism. Then $\tilde{\tau}:=\tau\oplus (-\tau^*): \g\ltimes_{\ad^*} \g^*\to \g\ltimes_{\ad^*} \g^*$ is a solution of the classical reflection equation of the Lie bialgebra $(\g\ltimes_{\ad^*} \g^*,r_B)$ given in Proposition \ref{r} if and only if $\tau$ is a reflection on the Rota-Baxter Lie algebra  $(\g,B)$.
\end{thm}
\begin{proof}
Let $\tau$ be a reflection on a  Rota-Baxter Lie algebra $(\g,B)$. By Theorem \ref{quathm}, we only need to show
that $\tilde{\tau}$ is a reflection on the quadratic Rota-Baxter Lie algebra $(\g\ltimes_{\ad^*} \g^*, B\oplus (-\lambda\textbf{1}_{\g^*}-B^*),\huaS)$.

First, following from $\tau[\tau x,y]_\g=[x,\tau y]_\g$, we have $\tau\circ \ad_{\tau x}=\ad_x\circ \tau$ and thus
\[[\tilde{\tau} x,\tilde{\tau}\xi]_{\g\ltimes_{\ad^*} \g^*}=-\ad^*_{\tau x}\tau^*\xi=-\tau^*\ad_x^*\xi=\tilde{\tau}[x,\xi]_{\g\ltimes_{\ad^*} \g^*},\quad x\in \g, \xi\in \g^*.\] Together with the fact that $\tau: \g\to \g$ is a Lie algebra automorphism, we see that $\tilde{\tau}$ is a Lie algebra automorphism.

Then, \eqref{tauinveq} is obvious from the definition of $\tilde{\tau}$.

At last, the relation  \eqref{ref1eq} holds for $\tilde{\tau}$ if and only if \eqref{involution} and \eqref{ref} are satisfied. In fact, the left hand side of $\eqref{ref1eq}$ in this case  is equivalent to
\begin{eqnarray*}
&&(\tilde{\tau}-\textbf{1}_{\g\ltimes_{\ad^*} \g^*})\circ (\tilde{B}+\lambda \textbf{1}_{\g\ltimes_{\ad^*} \g^*})\circ (\tilde{\tau}+\textbf{1}_{\g\ltimes_{\ad^*} \g^*})\\ &=&
\big((\tau-\textbf{1}_\g)\circ (B+\lambda \textbf{1}_\g)\circ (\tau+\textbf{1}_\g)\big) \oplus \big((-\tau^*-\textbf{1}_{\g^*})\circ (-\lambda \textbf{1}_{\g^*}-B^*+\lambda \textbf{1}_{\g^*})\circ (-\tau^*+\textbf{1}_{\g^*})\big)\\ &=&
\big((\tau-\textbf{1}_\g)\circ (B+\lambda \textbf{1}_\g)\circ (\tau+\textbf{1}_\g)\big) \oplus -\big((\tau-\textbf{1}_\g)\circ B\circ (\tau+\textbf{1}_\g)\big)^*\\
&=&(\tau \circ B\circ \tau-B+\tau\circ B-B\circ \tau+\lambda \tau^2-\lambda \textbf{1}_\g)\oplus-(\tau \circ B\circ \tau-B+\tau\circ B-B\circ \tau)^*,\end{eqnarray*}
which is equal to zero if and only if \eqref{involution} and \eqref{ref} hold. The other direction is easily obtained from the above argument.
\end{proof}
Theorem \ref{rrr} and Theorem \ref{Theorem 1} together imply the following result.
\begin{cor}\label{corRBcor}
Let $\tau$ be a reflection on the Rota-Baxter Lie algebra $(\g,B)$. Denote by $\h\subset \g$ and $\frkl\subset \g^*$ the fixed point sets of $\tau$ and $-\tau^*$, respectively. Then $\h\oplus \frkl$ is a coideal subalgebra of $(\g\ltimes_{\ad^*} \g^*, r_B)$.
\end{cor}

\begin{cor}\label{RBcor} Let $ \tau$ be a reflection on a  Rota-Baxter Lie algebra $(\g,B)$. Denote by $\h\subset \g$ the fixed points of $\tau$.
Then \[\mathrm{Im} (\tau+\iii_\g)\oplus \h^\perp \subset (\g\ltimes_{\ad^*} \g^*)_{B\oplus (-\lambda \iii_{\g^*}-B^*)}=\g_B\ltimes_{\rho_B} \g^*\] is a Rota-Baxter Lie subalgebra. 
\end{cor}
\begin{proof}
By Theorem \ref{rrr} and Proposition \ref{h^0},   the annihilator space of the fixed point set of
 $\tilde{\tau}=\tau\oplus (-\tau^*)$, which is \[\ker(\tilde{\tau}-\textbf{1}_{\g\ltimes_{\ad^*} \g^*})^\perp=\ker(\tau^*+\textbf{1}_{\g^*})^\perp\oplus \ker(\tau-\textbf{1}_\g)^\perp =\mathrm{Im}(\tau+\textbf{1}_\g)\oplus \h^\perp,\]
 is a Rota-Baxter Lie subalgebra of $\big(\g_B\ltimes_{\rho_B} \g^*, B\oplus (-\lambda \iii_{\g^*}-B^*)\big)$.
 \end{proof}
\begin{rmk}
If the reflection is an involution, then we have $\frkl=\h^\perp$ and $\mathrm{Im}(\tau+\iii_\g)=\h$. Therefore,  $\h\oplus \h^\perp$ is a coideal subalgebra of $(\g\ltimes_{\ad^*} \g^*, r_B)$ and a Rota-Baxter Lie subalgebra of $(\g_B\ltimes_{\rho_B} \g^*, B\oplus (-\lambda \iii_{\g^*}-B^*))$.
\end{rmk}

\begin{ex}
Any Lie algebra $\g$ with the negative identity map $-\iii_\g: \g\to \g$ is a Rota-Baxter Lie algebra of weight $1$. By Proposition \ref{r}, it generates a factorizable Lie bialgebra $(\g\ltimes_{\ad^*} \g^*, r_{-\iii_\g})$ with $r_{-\iii_\g}=\sum_{i} e_i\otimes e_i^*$, where $\{e_1,\cdots,e_n\}$ is a basis of $\g$ and $\{e_1^*,\cdots e_n^*\}$ is its dual basis. In fact, $(\g\ltimes_{\ad^*} \g^*, r_{-\iii_\g})$ is the double Lie bialgebra of the Lie bialgebra $(\g,\delta=0)$. In this situation, we have
\[(\g\oplus \g^*)_{r_{-\iii_\g}}=\g_{-\iii_\g}\ltimes_{\rho_{-\iii_\g}} \g^*=\overline{\g}\oplus \g^*,\]
where $  \overline{\g}=(\g,-[\cdot,\cdot]_\g)$ is the opposite Lie algebra.

Reflections on the Rota-Baxter Lie algebra $(\g,-\iii_\g)$ are precisely involutive automorphisms $\tau: \g\to \g$. Given an involutive automorphism $\tau$, the map $\tau\oplus (-\tau^*)$ is a solution of the classical reflection equation in the Lie bialgebra $(\g\ltimes_{\ad^*} \g^*, r_{-\iii_\g})$. Denote by $\h$ the fixed point set of $\tau$. Then
$\h\oplus \h^\perp$ is a coideal subalgebra of $(\g\ltimes_{\ad^*} \g^*, r_{-\iii_\g})$ and a Rota-Baxter Lie subalgebra of the direct sum Rota-Baxter Lie algebra $(\overline{\g}\oplus \g^*, -\iii_{\g}\oplus 0_{\g^*})$.

In particular, if $\g$ is a complex Lie algebra and the involution $\tau: \g\to \g$ is given by$ x\mapsto \overline{x}$, then $\h$ is the real form of $\g$.
\end{ex}

\section{Reflections on relative Rota-Baxter Lie algebras}\label{sec:rel}

In this section, we study reflections on relative Rota-Baxter Lie algebras, by which we construct solutions of the classical reflection equation in the Lie bialgebras induced by the relative Rota-Baxter Lie algebras.

\subsection{Reflections on relative Rota-Baxter Lie algebras of weight $0$}

Let $\g$ be a Lie algebra  and $\rho: \g\to \End(V)$ be a representation. A linear map $T: V\to \g$ is called a {\bf relative Rota-Baxter operator} (of weight $0$) if it satisfies
\[T(\rho(Tu)v-\rho(Tv) u)=[Tu,Tv]_\g,\quad \forall u,v\in V.\]
A {\bf relative Rota-Baxter Lie algebra} consists of a Lie algebra $\g$, a representation $\rho: \g\to \End(V)$ and a relative Rota-Baxter operator $T: V\to \g$. We denote by $(V\xrightarrow{T}\g,\rho)$ a relative Rota-Baxter Lie algebra.

Given  a relative Rota-Baxter  Lie algebra  $(V\xrightarrow{T}\g,\rho)$, it is well-known that $V$ is associated with a {\bf descendent Lie algebra} structure $[\cdot,\cdot]_T:\wedge^2V\to V$ given by
\[[u,v]_T=\rho(Tu)(v)-\rho(Tv)(u),\quad \forall u, v \in V.\]
Denote this Lie algebra by $V_T$.
Moreover, there is a Lie algebra representation $\theta: V_T\to \End(\g)$ of $V_T$ on $\g$ given by \[\theta(u)(x)=[Tu,x]_\g+T(\rho(x)u),\quad \forall u\in V,x\in \g.\]
We thus have the semi-direct product Lie algebra $V_T\ltimes_{\theta^*} \g^*$. The representation $\theta$ plays fundamental roles in the study of the cohomology theory of relative Rota-Baxter operators. See \cite{TBGS} for more details.

Moreover, for a relative Rota-Baxter Lie algebra $(V\xrightarrow{T}\g,\rho)$, treating $T\in V^*\otimes \g$, it was shown in \cite{Bai} that the skew-symmetrization $r_T=T-T^{21}\in \wedge^2(\g\oplus V^*)$ of $T$ is a solution of the classical Yang-Baxter equation in the semi-direct product Lie algebra $\g\ltimes_{\rho^*} V^*$. This leads to a triangular Lie bialgebra $(\g\ltimes_{\rho^*} V^*, r_T)$.

\begin{rmk}
We give an alternative approach to understand the solution $r_T=T-T^{21}$ of the classical Yang-Baxter equation for the Lie algebra $\g\ltimes_{\rho^*} V^*$ given by a relative Rota-Baxter Lie algebra $(V\xrightarrow{T}\g,\rho)$.
Introduce
 \[\overline{T}: \g^*\oplus V\to \g\ltimes_{\rho^*} V^*,\quad \overline{T}(\xi+u)=-T^*\xi+Tu,\quad \forall \xi\in \g^*, u\in V.\]
 Then $\overline{T}$ is a skew-symmetric relative Rota-Baxter operator on the Lie algebra $\g\ltimes_{\rho^*} V^*$ with respect to the coadjoint action on its dual $\g^*\oplus V$. Therefore,  the element $r\in \wedge^2 (\g\ltimes_{\rho^*} V^*)$,  determined by  $r_+:=\overline{T}$,  yields a skew-symmetric solution  of the classical Yang-Baxter equation, specifically  $r=T-T^{21}$.
\end{rmk}

  \begin{pro}\label{weight0case}
 The induced Lie bracket $[\cdot,\cdot]_{r_T}$ on the dual space $V\oplus \g^*$ in the triangular Lie bialgebra $(\g\ltimes_{\rho^*} V^*, r_T)$ is the same with the one on $V_T\ltimes_{\theta^*} \g^*$. In other words, the pair  $(\g\ltimes_{\rho^*} V^*, V_T\ltimes_{\theta^*} \g^*)$ of Lie algebras composes a triangular Lie bialgebra.
  \end{pro}
  \begin{proof}
We show that the Lie bracket $[a,b]_{r_T}=\add_{r_{T+} a}^*b-\add_{r_{T+} b}^*a$ on the dual space $V\oplus \g^*$ is the same with the one on $V_T\ltimes_{\theta^*} \g^*$. Here $\add^*$ is the coadjoint action of $\g\ltimes_{\rho^*} V^*$ on its dual and $r_{T+}(a)=r_T(a,\cdot)$. Indeed, by definition, for $u\in V$ and $\xi\in \g^*$, we have $r_{T+}(u+\xi)=Tu-T^*\xi$. Therefore, for $x+\alpha\in \g\oplus V^*$, one obtains
\begin{eqnarray*}
\langle [u,v+\xi]_{r_T}, x+\alpha\rangle &=&\langle \add_{Tu}^* (v+\xi)-\add_{Tv-T^*\xi}^* u, x+\alpha\rangle \\ &=&
-\langle  v+\xi, [Tu, x+\alpha]_{\g\ltimes_{\rho^*} V^*}\rangle +\langle  u,  [Tv-T^*\xi, x+\alpha]_{\g\ltimes_{\rho^*} V^*}\rangle\\ &=&
-\langle v, \rho^*(Tu)\alpha\rangle +\langle u, \rho^*(Tv)\alpha\rangle-\langle \xi, [Tu,x]_\g\rangle+\langle u, \rho^*(x) T^*\xi\rangle\\ &=&
\langle \rho(Tu)(v)-\rho(Tv)(u), \alpha\rangle-\langle \xi,[Tu,x]_\g+T(\rho(x)u)\rangle\\ &=&
\langle \rho(Tu)(v)-\rho(Tv)(u), \alpha\rangle+\langle \theta^*(u)(\xi),x\rangle,\end{eqnarray*}
which implies that
\[[u,v]_{r_T}=[u,v]_T\in V,\qquad [u,\xi]_{r_T}=\theta^*(u)(\xi)\in \g^*.\]
 For all $\xi,\eta\in \g^*$, we have
\begin{eqnarray*}
\langle [\xi,\eta]_{r_T},x+\alpha\rangle&=&\langle -\add_{T^*\xi}^*\eta+\add_{T^*\eta}^* \xi, x+\alpha\rangle\\
&=&\langle\eta, [T^*\xi,x+\alpha]_{\g\ltimes_{\rho^*} V^*}\rangle-\langle \xi,[T^*\eta,x+\alpha]_{\g\ltimes_{\rho^*} V^*}\rangle\\
&=&0,
\end{eqnarray*}
 which implies that  $[\xi,\eta]_{r_T}=0$. Therefore, the Lie bracket $[\cdot,\cdot]_{r_T}$ on the dual $V\oplus \g^*$ in the triangular Lie bialgebra $(\g\ltimes_{\rho^*} V^*, r_T)$ is the same with the semidirect product $V_T\ltimes_{\theta^*} \g^*$.
 \end{proof}

\begin{defi}\label{refzero}
A {\bf reflection} on a relative Rota-Baxter Lie algebra $(V\xrightarrow{T}\g,\rho)$ consists  of a linear isomorphism $\tau:V\to V$ and a Lie algebra automorphism $\sigma: \g\to \g$ such that
\begin{eqnarray}
\label{eq1}\tau(\rho(\sigma x)v)-\rho(x)(\tau v)&=&0,\quad \forall x\in \g, v\in V,\\
\label{eq2}\sigma\circ T\circ \tau-T+\sigma\circ T-T\circ \tau&=&0.
\end{eqnarray}
\end{defi}

\begin{thm}\label{thm:rRB-sol}
Let  $(\tau,\sigma)$ be a reflection on  a relative Rota-Baxter Lie algebra $(V\xrightarrow{T}\g,\rho)$. Then $\tilde{\tau}:=\sigma\oplus (-\tau^*): \g\ltimes_{\rho^*} V^*\to  \g\ltimes_{\rho^*} V^*$ is a solution of the classical reflection equation in the triangular Lie bialgebra $(\g\ltimes_{\rho^*} V^*,r_T)$, where $r_T=T-T^{21}$.
 \end{thm}
\begin{proof}
Following from \eqref{eq1} and the fact that $\sigma$ is a Lie algebra automorphism, we see that
 $\tilde{\tau}$ is a Lie algebra automorphism. It is left to show the map
 \[(\tilde{\tau}-\textbf{1}_{\g\oplus V^*})\circ r_{T+}\circ (\tilde{\tau}^*-\textbf{1}_{\g^*\oplus V}): \g^*\oplus V\to \g\oplus V^*\]
 is zero. In fact, we have
 \begin{eqnarray*}
&& (\tilde{\tau}-\textbf{1}_{\g\oplus V^*})\circ r_{T+}\circ (\tilde{\tau}^*-\textbf{1}_{\g^*\oplus V})\\
 &=&\big((\sigma-\textbf{1}_\g)\oplus (-\tau^*-\textbf{1}_{V^*})\big)\circ \big((T\oplus (-T^*))\circ (-\tau-\textbf{1}_V)\oplus (\sigma^*-\textbf{1}_{\g^*})\big)\\ &=&\big(-(\sigma-\textbf{1}_\g)\circ T\circ (\tau+\textbf{1}_V)\big)\oplus \big((\sigma-\textbf{1}_\g)\circ T\circ (\tau+\textbf{1}_V)\big)^*\\ &=&0,
    \end{eqnarray*}
 where we used \eqref{eq2} in the last equation.
\end{proof}

\begin{cor}\label{relweight0}
Let  $(\tau,\sigma)$ be a reflection on  a relative Rota-Baxter Lie algebra $(V\xrightarrow{T}\g,\rho)$. Denote by $\h$ and $\frkl$ the fixed point sets of $\sigma$ and $-\tau^*$, respectively. Then $\h\oplus \frkl$ is a coideal subalgebra of $(\g\ltimes_{\rho^*} V^*, r_T)$.
\end{cor}
\begin{proof}
Combining  Theorem \ref{thm:rRB-sol} and  Theorem \ref{Theorem 1}, we see that the $\frkl^\perp\oplus \h^\perp \subset (\g\oplus V^*)_{r_{T}}$ is a Lie subalgebra so that $\h\oplus \frkl$ is a coideal subalgebra.
\end{proof}

 \begin{thm}\label{perpzero}
Let  $(\tau,\sigma)$ be a reflection on  a relative Rota-Baxter Lie algebra $(V\xrightarrow{T}\g,\rho)$. Denote by $\h$ the space of fixed points of $\sigma$. Then
\[\mathrm{Im}(\tau+\iii_V)\oplus \h^\perp \subset V_T\ltimes_{\theta^*} \g^*\]
is a Lie subalgebra.
\end{thm}

\begin{proof}
Note that the annihilator space of the fixed point set of $\sigma\oplus (-\tau^*)$ in Theorem \ref{thm:rRB-sol} is
\[\ker(\tilde{\tau}-\iii_{\g\ltimes_{\rho^*} V^*})^\perp=\ker(-\tau^*-\iii_{V^*})^\perp\oplus \ker(\sigma-\iii_\g)^\perp=\mathrm{Im}(\tau+\iii_V)\oplus \h^\perp.\]
It is a Lie subalgebra by Theorem \ref{thm:rRB-sol} and Corollary \ref{relweight0}. But we could also prove it directly as follows.

First, by \eqref{eq2}, for any $u\in V$, we have $\sigma T(\tau u+u)=Tu+T\tau u$. This implies that
\begin{eqnarray}\label{IMT}
\mathrm{Im}(T\circ (\tau+\textbf{1}_V))\subset \h.
\end{eqnarray}
For $\xi\in \h^\perp$, we check that $\theta^*(\tau u+u)(\xi)\in \h^\perp$. Indeed, for $x\in \h$, by \eqref{eq1}, we have
\begin{eqnarray*}
\langle \theta^*(\tau u+u)(\xi),x\rangle&=&-\langle \xi, [T(\tau u+u),x]+T(\rho(x)(\tau u+u))\rangle\\
&=&-\langle \xi, [T(\tau u+u),x]+T(\tau(\rho(x)u)+\rho(x)u)\rangle\\ &=&0,
\end{eqnarray*}
where we have used the fact that  $\h\subset \g$ is a Lie subalgebra. This proves that $\theta^*(\tau u+u)(\xi)\in \h^\perp$.
It suffices to show $ \mathrm{Im}(\tau+\textbf{1}_V)\subset V_T$ is a Lie subalgebra. Indeed, for any $u, v\in V$, by \eqref{IMT} and \eqref{eq1}, we have
\begin{eqnarray*}
&&[\tau u+u,\tau v+v]_T\\&=&\rho(T(\tau u+u))(\tau v+v)-\rho(T(\tau v+v))(\tau u+u)\\ &=&\tau\big(\rho(\sigma T(\tau u+u))v\big)+\rho(T(\tau u+u))(v)-\tau\big(\rho(\sigma T(\tau v+v))u\big)-\rho(T(\tau v+v))(u)\\ &=&\tau\big(\rho(T(\tau u+u))v\big)+\rho(T(\tau u+u))(v)-\tau\big(\rho(T(\tau v+v))u\big)-\rho(T(\tau v+v))(u)\\ &=& \tau w+w,
\end{eqnarray*}
where $w=\rho(T(\tau u+u))(v)-\rho(T(\tau v+v))(u)$. This completes the proof.
\end{proof}

It is well-known that pre-Lie algebras are the underlying algebraic structures of relative Rota-Baxter operators of weight 0. At the end of this subsection, we use pre-Lie algebras to construct solutions of the classical reflection equation. Recall that  a {\bf pre-Lie algebra} is a pair $(\g,\triangleright)$, where $\g$ is a vector space and  $\triangleright:\g\otimes \g\longrightarrow \g$ is a bilinear multiplication
satisfying that for all $x,y,z\in \g$, the associator
$$(x,y,z):=(x\triangleright y)\triangleright z-x\triangleright(y\triangleright z)
$$
is symmetric in $x,y$, that is,
\vspace{-.1cm}
$$(x,y,z)=(y,x,z)\;\;{\rm or}\;\;{\rm
equivalently,}\;\;(x\triangleright y)\triangleright z-x\triangleright(y\triangleright z)=(y\triangleright x)\triangleright
z-y\triangleright(x\triangleright z).$$

Associated to a pre-Lie algebra $(\g,\triangleright)$, we have the subadjacent Lie algebra $\g_{\triangleright}$ with Lie bracket given by $[x,y]_\triangleright=x\triangleright y-y\triangleright x$, which has a representation on $\g$ given by
\[L_\triangleright: \g_{\triangleright}\to \End(\g),\qquad {L_\triangleright}_xy=x\triangleright y.\]
Moreover, the identity map $\iii_\g:\g\to\g_\triangleright$ is a relative Rota-Baxter operator on the subadjacent Lie algebra $\g_\triangleright$ with respect to the representation $L_\triangleright$.

Conversely, given a relative Rota-Baxter Lie algebra $(V\xrightarrow{T} \g,\rho)$, we have a pre-Lie structure on $V$ given by
$u\triangleright_T v:=\rho(Tu)v$.

Let $(\g,\triangleright)$ be a pre-Lie algebra. Since the identity map $\iii_\g:\g\to\g_\triangleright$ is a relative Rota-Baxter operator, it follows that
\begin{eqnarray}\label{eq:rrr}
  r_{\iii_\g}:=\sum^{n}_{i=1}(e_i\otimes e_i^*-e_i^*\otimes e_i)
\end{eqnarray}
is a skew-symmetric solution of the classical Yang-Baxter equation in the semi-direct product  Lie algebra $\g_\triangleright\ltimes_{L_\triangleright^*}\g^*,$ where $\{e_1,\cdots,e_n\}$ is a basis of $\g$ and $\{e^*_1,\cdots,e^*_n\}$ is its dual basis (\cite{Bai}).

\begin{pro}
Let $(\g,\triangleright)$ be a pre-Lie algebra and $\sigma: \g\to \g$ an involutive pre-Lie algebra automorphism. Then $(\sigma,\sigma)$ is a reflection on the corresponding relative Rota-Baxter Lie algebra $(\g \xrightarrow{\iii_{\g}} \g_{\triangleright}, L_\triangleright)$, and $\sigma\oplus (-\sigma^*)$ is a solution of the classical reflection equation in the triangular Lie bialgebra $(\g_\triangleright\ltimes_{L_\triangleright^*}\g^*,r_{\iii_\g})$, where $r_{\iii_\g}$ is given by \eqref{eq:rrr}.
\end{pro}
\begin{proof}
Observing that in this circumstance, i.e. $\sigma=\tau$ and $T=\iii_\g$, the relation \eqref{eq2} holds if and only if $\sigma^2=\iii_\g$. Under this condition, \eqref{eq1} is equivalent to $\sigma (x\triangleright y)=(\sigma^{-1} x)\triangleright (\sigma y)=(\sigma x)\triangleright (\sigma y)$. Namely, $\sigma$ is a pre-Lie algebra homomorphism.  This implies that $\sigma:\g_\triangleright \to \g_\triangleright$ is a Lie algebra automorphism. Indeed, $(\sigma,\sigma)$ is a reflection on the relative Rota-Baxter Lie algebra if and only if $\sigma: \g\to \g$ is an involutive pre-Lie algebra automorphism. By Theorem \ref{thm:rRB-sol},  $\sigma\oplus (-\sigma^*)$ is a solution of the classical reflection equation in the triangular Lie bialgebra $(\g_\triangleright\ltimes_{L_\triangleright^*}\g^*,r_{\iii_\g})$.
\end{proof}


\subsection{Reflections on relative Rota-Baxter Lie algebras of nonzero weight}
Let $\g,\frkk$ be two Lie algebras and $\rho:\g\to \Der(\frkk)$ be a representation of $\g$ on $\frkk$ by derivations. A linear map $T: \frkk\to \g$ is called a {\bf relative Rota-Baxter operator} of weight $\lambda$ ($\lambda\neq 0$) if it satisfies
\[ [Tu,Tv]_\g=T(\rho(Tu)v-\rho(Tv)u+\lambda[u,v]_{\frkk}),\quad \forall u,v\in \frkk.\]
A relative Rota-Baxter Lie algebra (of nonzero weight) consists of a Lie algebra  $\g$, a representation $\rho:\g\to \Der(\frkk)$ and a relative Rota-Baxter operator $T: \frkk\to \g$ of weight $\lambda$, which is denoted by $(\frkk\xrightarrow{T}\g,\rho)$. Associated to a relative Rota-Baxter Lie algebra $(\frkk\xrightarrow{T}\g,\rho)$, we have a semi-direct product Lie algebra $\overline{\g}:=\g\ltimes_{\rho}\frkk$. Define
\begin{eqnarray*}
\overline{T}: \g\ltimes_{\rho}\frkk\to \g\ltimes_{\rho}\frkk,\qquad \overline{T}(x+u)=-\lambda x+Tu,\quad \forall x\in \g,u\in \frkk.
\end{eqnarray*}
It was proved in \cite[Proposition 4.1]{BNG}
 that $\overline{T}$ is a Rota-Baxter operator of weight $\lambda$ on the Lie algebra $\overline{\g}:=\g\ltimes_{\rho}\frkk$. Consequently, $\big(\overline{\g}\ltimes_{\ad^*} \overline{\g}^*,\overline{T}\oplus (-\lambda\iii-\overline{T}^*),\huaS\big)$ is a quadratic Rota-Baxter Lie algebra, where $\huaS$ is given by
$$
\huaS(x+u+\xi+\alpha,y+v+\eta+\beta)=\langle x,\eta\rangle+\langle u,\beta\rangle+\langle \xi, y\rangle+\langle \alpha,v\rangle,
$$where $x,y\in \g, u,v\in \frkk, \xi,\eta\in \g^*$ and $\alpha,\beta\in \frkk^*$.

By Proposition \ref{r}, we have the following corollary.
\begin{cor}\label{cor:explicit-formula}
  Let $(\frkk\xrightarrow{T} \g,\rho)$ be a relative Rota-Baxter Lie algebra of nonzero weight $\lambda$. Then
 $r_{\overline{T}}\in \otimes^2 (\overline{\g}\ltimes_{\ad^*} \overline{\g}^*)$ defined by
\[r_{\overline{T}}(x+u+\xi+\alpha,y+v+\eta+\beta)=\frac{1}{\lambda}\big(\langle Tu, \eta\rangle-\langle Tv,\xi\rangle+\lambda\langle \xi,y\rangle+\lambda\langle \beta,u\rangle\big)\]
is a solution of the classical Yang-Baxter equation in the Lie algebra $\overline{\g}\ltimes_{\ad^*} \overline{\g}^*$.

Furthermore, the induced Lie algebra structure $[\cdot,\cdot]_{r_{\overline{T}}}$ on the dual space $(\overline{\g}\ltimes_{\ad^*} \overline{\g}^*)^*=\g\oplus \frkk\oplus \g^*\oplus \frkk^*$ is:
\begin{eqnarray*}
{[x+u,y+v]_{r_{\overline{T}}}}&=&\frac{1}{\lambda} \big(-\lambda[x,y]_\g+[x,Tv]_\g+[Tu,y]_\g+\rho(Tu)v-\rho(Tv)u+\lambda[u,v]_\frkk)\big);\\
{[x+u,\xi+\alpha]_{r_{\overline{T}}}}&=&\frac{1}{\lambda}\big(\ad_{Tu}^*\xi+(\rho_{-\lambda x+Tu}^*\alpha-\ad^*_u T^*\xi-\rho_x^*T^*\xi)\big);\\
{[\xi+\alpha,\eta+\beta]_{r_{\overline{T}}}}&=&0.\end{eqnarray*}
\end{cor}
\begin{proof}
  By (ii) in Proposition \ref{r}, we obtain
 \begin{eqnarray*}
 r_{\overline{T}}(x+u+\xi+\alpha,y+v+\eta+\beta)&=&\frac{1}{\lambda}\big(\langle \overline{T}(x+u), \eta+\beta\rangle-\langle \xi+\alpha,\overline{T}(y+v)\rangle+\lambda \langle x+u, \eta+\beta\rangle\big)\\&=&\frac{1}{\lambda}\big(\langle Tu, \eta\rangle-\langle Tv,\xi\rangle+\lambda\langle \xi,y\rangle+\lambda\langle \beta,u\rangle\big).
  \end{eqnarray*}
Based on Lemma \ref{iso}, we see that
\[\frac{1}{\lambda}\iii: (\g\oplus \frkk \oplus \g^*\oplus \frkk^*)_{r_{\overline{T}}}\to \big(\overline{\g}\ltimes_{\ad^*} \overline{\g}^*\big)_{\overline{T}\oplus (-\lambda \iii_{\overline{\g}^*}-\overline{T}^*)},\]
is an isomorphism from the dual Lie algebra induced by $r_{\overline{T}}$ to the descendent Lie algebra of the Rota-Baxter operator $\overline{T}\oplus (-\lambda \iii_{\overline{\g}^*}-\overline{T}^*)$. Using this isomorphism, it is straightforward to obtain the explicit formulas given above.
\end{proof}

Recall from Proposition \ref{weight0case} that  a relative Rota-Baxter Lie algebra of weight $0$ induces a triangular Lie bialgebra $(\g\ltimes_{\rho^*}V^*, \g^*\rtimes V_T)$.  In this following, we also find a  Lie bialgebra $(\g\ltimes_{\rho^*} \frkk^*, \g^*\rtimes \frkk_T)$ for the nonzero weight case.  Although this Lie bialgebra  is not coboundary, it is a Lie subalgebra of the  factorizable Lie bialgebra in Corollary \ref{cor:explicit-formula}.

A Lie bialgebra $(\h,\h^*)$ is called a {\bf Lie sub-bialgebra} of a Lie bialgebra $(\g,\g^*)$, if $\h\subset \g$ is a Lie subalgebra and the projection $\g^*\to \h^*$ is a Lie algebra homomorphism. The latter is equivalent to say that $\delta(\h)\subset \h\wedge \h$; see \cite{LPV}.

\begin{pro}\label{main}
Let $(\frkk\xrightarrow{T} \g,\rho)$ be a relative Rota-Baxter  Lie algebra of  weight $\lambda$. Then the associated factorizable Lie bialgebra  $\big((\g\ltimes_{\rho} \frkk)\ltimes_{\ad^*} (\g^*\oplus \frkk^*),r_{\overline{T}}\big)$  has  a Lie sub-bialgebra $(\g\ltimes_{\rho^*} \frkk^*, \g^*\rtimes \frkk_T)$, where the Lie bracket on $\g^*\rtimes \frkk_T$ is given by
\begin{eqnarray}\label{subbi}
[\xi+u,\eta+v]_*=\frac{1}{\lambda}(\ad^*_{Tu}\eta-\ad^*_{Tv}\xi+\rho(Tu)v-\rho(Tv)u+\lambda[u,v]_{\frkk}),
\end{eqnarray}
for $\xi,\eta\in \g^*$ and $u,v\in \frkk$.
\end{pro}
\begin{proof}
By definition, we see that
\[\g\ltimes_{\rho^*} \frkk^*\subset (\g\ltimes_{\rho} \frkk)\ltimes_{\ad^*} (\g^*\oplus \frkk^*)\]
is a Lie subalgebra.
Using the explicit formula presented in Corollary \ref{cor:explicit-formula},   it is straightforward to check that
\[(\g\ltimes_{\rho^*} \frkk^*)^\perp=\g\oplus \frkk^*\subset (\g\oplus \frkk \oplus \g^*\oplus \frkk^*)_{r_{\overline{T}}}\]
is an ideal.
Equipped $\g^*\rtimes \frkk_T=(\g\oplus \frkk \oplus \g^*\oplus \frkk^*)_{r_{\overline{T}}}/(\g\ltimes_{\rho^*} \frkk^*)^\perp$  with the quotient Lie algebra structure, we see that the projection
\[(\g\oplus \frkk \oplus \g^*\oplus \frkk^*)_{r_{\overline{T}}}\to \g^*\rtimes \frkk_T\]
is a Lie algebra homomorphism. Therefore, $(\g\ltimes_{\rho^*} \frkk^*, \g^*\rtimes \frkk_T)$ is a Lie sub-bialgebra.
\end{proof}
\begin{rmk}In \cite[Corollary 3.8 (b)]{BNG},  for a relative Rota-Baxter Lie algebra $(\frkk\xrightarrow{T} \g,\rho)$ of weight $\lambda$, the authors proved that  $r_T:=T-T^{21}\in \wedge^2 (\g\ltimes_{\rho^*}\frkk^*)$ is a skew-symmetric solution of the generalized Yang-Baxter equation: $[x, [r_T,r_T]]=0$ for all $x\in \g\ltimes_{\rho^*}\frkk^*$ if and only if two conditions are satisfied.
A coboundary Lie bialgebra $(\g\ltimes_{\rho^*} \frkk^*, (\g^*\oplus \frkk)_{r_T})$ is thereby induced. By direct calculation, the Lie bracket $[\cdot,\cdot]_{r_T}$ on $(\g^*\oplus \frkk)_{r_T}$ is
\[[\xi+u,\eta+v]_{r_T}=\ad_{Tu}^*\eta-\ad_{Tv}^*\xi+\langle \xi, T(\rho(\cdot)(v))\rangle-\langle \eta, T(\rho(\cdot)(u))\rangle+\rho(Tu)v-\rho(Tv)u.\]
Actually, the two extra conditions are needed to make $[\cdot,\cdot]_{r_T}$  a Lie bracket.

We find a  different Lie bracket  \eqref{subbi} on $\g^*\oplus \frkk$ so that $(\g\ltimes_{\rho^*} \frkk^*, \g^*\rtimes \frkk_T)$ is a Lie bialgebra automatically. Although this Lie bialgebra is not  coboundary, it is a Lie sub-bialgebra of a bigger factorizable Lie bialgebra $(\overline{\g}\ltimes_{\ad^*} \overline{\g}^*, r_{\overline{T}})$.
\end{rmk}

\begin{defi}\label{refrel}
A {\bf reflection} on a relative Rota-Baxter Lie algebra $(\frkk\xrightarrow{T}\g, \rho)$ of nonzero weight $\lambda$  consists of two involutive Lie algebra automorphisms  $\tau: \frkk\to \frkk$ and $\sigma: \g\to \g$ such that
\begin{eqnarray}
\label{req1}\tau(\rho(x)u)-\rho(\sigma x) (\tau u)&=&0, \\
\label{req5}\sigma \circ T\circ \tau-T+\sigma\circ T-T\circ \tau&=&0.
\end{eqnarray}
\end{defi}
\begin{thm}\label{rRBth}
Let $(\sigma,\tau)$ be a reflection on a relative Rota-Baxter Lie algebra $(\frkk\xrightarrow{T}\g, \rho)$ of weight $\lambda$. Then \[(\sigma\oplus \tau)\oplus (-\sigma^*\oplus -\tau^*): \overline{\g}\ltimes_{\ad^*} \overline{\g^*}\to \overline{\g}\ltimes_{\ad^*} \overline{\g^*}\] is a solution of the classical reflection equation of the Lie bialgebra $(\overline{\g}\ltimes_{\ad^*} \overline{\g^*}, r_{\overline{T}})$, where $\overline{\g}=\g\ltimes_\rho \frkk$ and $\overline{T}(x+u)=-\lambda x+Tu$.
\end{thm}
\begin{proof}
We only need to prove that $\tilde{\tau}:=\sigma\oplus \tau$ is a reflection on the Rota-Baxter Lie algebra $\big(\overline{\g},\overline{T})$ as introduced in Definition \ref{reflectiononRB}. Then the result follows from Theorem \ref{rrr}.

First, based on the fact that $\sigma$ and $\tau$ are  Lie algebra automorphisms and the relation \eqref{req1}, we see that $\tilde{\tau}: \overline{\g}\to \overline{\g}$ is a Lie algebra automorphism.
Then  the condition \eqref{homo}, i.e., $\tilde{\tau}[\tilde{\tau} a,b]=[a,\tilde{\tau} b]$ follows from \eqref{req1} and the facts that $\sigma, \tau$ are  Lie algebra homomorphisms and involutions. Condition \eqref{involution} holds since $\sigma$ and $\tau$ are involutions.  It is left to show
\eqref{ref}, namely \[\tilde{\tau} \circ \overline{T}\circ \tilde{\tau}-\overline{T}+\tilde{\tau}\circ \overline{T}-\overline{T}\circ \tilde{\tau}=(\tilde{\tau}-\textbf{1}_{\overline{\g}})\circ \overline{T}\circ (\tilde{\tau}+\textbf{1}_{\overline{\g}})=0.\]
Indeed, we have
\begin{eqnarray*}
(\tilde{\tau}-\textbf{1}_{\overline{\g}})\circ \overline{T}\circ (\tilde{\tau}+\textbf{1}_{\overline{\g}})(x+u)&=&
(\tilde{\tau}-\textbf{1}_{\overline{\g}})(-\lambda (\sigma+\textbf{1}_\g) (x)+T(\tau+\textbf{1}_\frkk)(u))\\ &=&
-\lambda (\sigma-\textbf{1}_\g)(\sigma+\textbf{1}_\g) (x)+(\sigma-\textbf{1}_\g)T(\tau+\textbf{1}_\frkk)(u) \\ &=&0,\end{eqnarray*}
where we used \eqref{req5} and $\sigma^2=\iii_{\g}$ in the last step. So we obtain that $\tilde{\tau}$ is a reflection on the Rota-Baxter Lie algebra $\big(\overline{\g},\overline{T})$.\end{proof}
\begin{cor}
Let $(\sigma,\tau)$ be a reflection on a relative Rota-Baxter operator $(\frkk\xrightarrow{T}\g, \rho)$ of weight $\lambda$.
Denote by $\h_\sigma\subset \g, \h_\tau\subset \frkk, \frkl_\sigma\subset \g^*, \frkl_\tau\subset \frkk^*$ the fixed point sets of $ \sigma, \tau, -\sigma^*, -\tau^*$, respectively.
Then
\begin{itemize}
\item[\rm (i)]$\h_\sigma \oplus \h_{\tau}\oplus \frkl_\sigma\oplus \frkl_\tau$ is a coideal subalgebra of $\big((\g\ltimes_\rho \frkk)\ltimes_{\ad^*} (\g^*\oplus \frkk^*),r_{\overline{T}}\big)$;
\item[\rm (ii)] $\mathrm{Im}(\sigma+\iii_\g)\oplus \mathrm{Im}(\tau+\iii_\frkk)\oplus \h_\sigma^\perp\oplus \h_\tau^\perp \subset (\g\rtimes \frkk_T)\ltimes (\g^*\oplus \frkk^*)$
is a Rota-Baxter Lie subalgebra.
\end{itemize}
\end{cor}
\begin{proof}
The results are deduced from Theorem \ref{rRBth} and Corollary \ref{corRBcor} and Corollary \ref{RBcor}.
\end{proof}

At the end of this subsection, we use involutive automorphisms on post-Lie algebras to construct reflections on relative Rota-Baxter Lie algebras of nonzero weight. The notion of a post-Lie algebra was introduced by Vallette from his study of Koszul duality of operads in \cite{Val}, and has been found various applications in Rota-Baxter algebras, numerical analysis and regularity structures.

\begin{defi} \cite{Val}\label{post-lie-defi}
A {\bf post-Lie algebra} $(\g,[\cdot,\cdot]_\g,\triangleright_\g)$ consists of a Lie algebra $(\g,[\cdot,\cdot]_\g)$ and a binary product $\triangleright_\g:\g\otimes\g\lon\g$ such that
\begin{eqnarray}
\label{Post-1}x\triangleright_\g[y,z]_\g&=&[x\triangleright_\g y,z]_\g+[y,x\triangleright_\g z]_\g,\\
\label{Post-2}([x,y]_\g+x\triangleright_\g y-y\triangleright_\g x) \triangleright_\g z&=&x\triangleright_\g(y\triangleright_\g z)-y\triangleright_\g(x\triangleright_\g z),
\end{eqnarray}
where $x, y, z\in\g$.
\end{defi}
A {\bf homomorphism} between two post-Lie algebras $(\g,[\cdot,\cdot]_\g,\triangleright_\g)$ and $(\h,[\cdot,\cdot]_\h,\triangleright_\h)$ is a linear map $\sigma: \g\to \h$ such that $\sigma[x,y]_\g=[\sigma x,\sigma y]_\h$ and $\sigma(x\triangleright_\g y)=(\sigma x)\triangleright_\h (\sigma y)$.

For a post-Lie algebra $(\g,[\cdot,\cdot]_\g,\triangleright_\g)$, it is obvious that if $[\cdot,\cdot]_\g=0$, then $(\g,\triangleright_\g)$ is a pre-Lie algebra.

For any post-Lie algebra $(\g,[\cdot,\cdot]_\g,\triangleright_\g)$, define a linear map $[\cdot,\cdot]_\triangleright:\g\wedge\g\to\g$ as follows:
\begin{eqnarray}\label{sub-adja-Lie}
&[x,y]_\triangleright=[x,y]_\g+x\triangleright_\g y-y\triangleright_\g x,\quad\forall x,y\in\g.
\end{eqnarray}
Then $(\g,[\cdot,\cdot]_\triangleright)$ is a Lie algebra, which is called the {\bf subadjacent Lie algebra}, and denoted by $\g_\triangleright$. Moreover, $L_\triangleright:\g_\triangleright\to\Der(\g)$ is an action of the Lie algebra $\g_\triangleright$ on the original Lie algebra $(\g,[\cdot,\cdot]_\g)$, and $\iii_\g:\g\to\g_\triangleright$ is a relative Rota-Baxter operator of weight 1.  Conversely, a relative Rota-Baxter operator of nonzero weight naturally induces a post-Lie algebra \cite{BGN2010}.

\begin{pro}
Let $(\g,[\cdot,\cdot]_\g, \triangleright)$ be a post-Lie algebra and $\sigma: \g\to \g$ an involutive post-Lie algebra automorphism. Then $(\sigma,\sigma)$ is a reflection on the corresponding relative Rota-Baxter Lie algebra $(\g \xrightarrow{\iii_{\g}} \g_{\triangleright}, L_\triangleright)$.
\end{pro}
\begin{proof}
 Since $\sigma$ is an involutive post-Lie algebra automorphism, we see that
\[ \sigma[x,y]_\triangleright=\sigma([x,y]_\g+x\triangleright_\g y-y\triangleright_\g x)=[\sigma x,\sigma y]_\g+(\sigma x)\triangleright_\g (\sigma y)-(\sigma y)\triangleright_\g (\sigma x)=[\sigma x,\sigma y]_\triangleright.\]
So $\sigma: \g_\triangleright\to \g_\triangleright$ and $\sigma: \g\to \g$ are both Lie algebra automorphisms. Moreover, \eqref{req1} follows from $\sigma(x\triangleright_\g y)=(\sigma x)\triangleright_\g(\sigma y)$ and \eqref{req5} is exactly $\sigma^2=\iii_\g$. So we see that $(\sigma,\sigma)$ is a reflection on $(\g \xrightarrow{\iii_{\g}} \g_{\triangleright}, L_\triangleright)$. \end{proof}

\section{Reflections on Rota-Baxter Lie algebras and Poisson homogeneous spaces}\label{sec:PL}

In this section, from Rota-Baxter operators and reflections on Rota-Baxter operators, we  construct explicit Poisson Lie groups and Poisson homogeneous spaces, respectively.

For a Lie group $G$ and $g\in G$, we use the notations $\huaL_g$ and $\huaR_g$ to denote the tangent maps of the left and right translations by $g$ at identity.

\begin{pro}\label{quagroup}
Let $G$ be a connected and simply-connected Lie group such that its Lie algebra $(\g,B,S)$ is a quadratic Rota-Baxter Lie algebra of weight $\lambda$. Then
\begin{itemize}
\item[\rm (i)] $(G,\pi)$ is a Poisson Lie group, where $\pi\in \Gamma(\wedge^2 TG)$ is defined by
\begin{eqnarray*}
\pi_g&=&(\huaL_{g}\otimes \huaL_g)(B\circ I_S)-(\huaR_{g}\otimes \huaR_{g})(B\circ I_S),\quad \lambda=0;\\
\pi_g&=&\frac{1}{2\lambda}\big((\huaL_{g}\otimes \huaL_{g})(B\circ I_S-(B\circ I_S)^{21})-(\huaR_{g}\otimes \huaR_{g})(B\circ I_S-(B\circ I_S)^{21})\big),\quad \lambda\neq 0,
\end{eqnarray*}
where  the map  $B\circ I_S:\g^*\to \g$ is viewed as an element in $\g\otimes \g$.
\item[\rm (ii)] if $\tau:\g\to \g$ is a reflection on the quadratic Rota-Baxter Lie algebra as introduced in Definition \ref{refqua}, then
$(G/H,\mathrm{pr}_* \pi)$ is a Poisson homogeneous space of the Poisson Lie group $(G,\pi)$, where $H\subset G$ is the closed Lie subgroup with the Lie algebra $\h=\ker(\tau-\iii_{\g})$.
\end{itemize}
\end{pro}

\begin{proof}
For (i), following from Theorem \ref{quathm}, associated to a quadratic Rota-Baxter Lie algebra $(\g,B,S)$ of arbitrary weight, there is always a Lie bialgebra $(\g, r_B)$. The corresponding Poisson Lie group is $(G,\pi=\overleftarrow{r_B}-\overrightarrow{r_B})$.
Denote by $t_B$ the skew-symmetric part of $r_B$. Then $\pi=\overleftarrow{t_B}-\overrightarrow{t_B}$. When $\lambda=0$, we have $t_B=r_B=B\circ I_S\in \wedge^2 \g$. When $\lambda\neq 0$, we have
\[t_{B+}=\frac{1}{2}(r_++r_-)=\frac{1}{2}(\frac{1}{\lambda} B\circ I_S+I_S-\frac{1}{\lambda}I_S \circ B^*-I_S)=\frac{1}{2\lambda}(B\circ I_S-I_S\circ B^*).\]
So $t_B=\frac{1}{2\lambda} (B\circ I_S-(B\circ I_S)^{21})$ and we obtain the formula of $\pi$ in this case.

For (ii), it is known that the quotient manifold $G/H$ is a Poisson homogeneous space of $(G,\pi)$ such that the projection $G\to G/H$ is a Poisson map if and only if the annihilator space $\h^\perp\subset \g^*$ is a Lie subalgebra, which amounts to that $\h^0=I_S(\h^\perp)\subset \g_B$ is a Lie subalgebra.
This is exactly the case by Proposition  \ref{h^0}.
\end{proof}

\begin{ex}
For the quadratic Rota-Baxter Lie algebra $(\mathrm{sl}(n,\mathbb{C})\oplus \mathrm{sl}(n,\mathbb{C}), \hat{B},\hat{S})$ in Example \ref{sln} with the  reflection $\tau(x,u)=(u,x)$, we have a Poisson Lie group $\mathrm{SL}(n,\mathbb{C})\times \mathrm{SL}(n,\mathbb{C})$ and a Poisson homogeneous space $(\mathrm{SL}(n,\mathbb{C})\times \mathrm{SL}(n,\mathbb{C}))/\mathrm{SL}(n,\mathbb{C})$. This Poisson homogeneous space  is actually  Poisson diffeomorphic to the dual Poisson Lie group $$\mathrm{SL}(n,\mathbb{C})^*=B_+* B_-:=\{(b,c)\in B_+\times B_-|b_{ii}c_{ii}=1\},$$ where $B_{\pm}$  consist of  upper and lower triangular matrices in $\mathrm{SL}(n,\mathbb{C})$.
\end{ex}
\begin{ex}
For the quadratic Rota-Baxter Lie algebra $(\mathrm{sl}(n,\mathbb{C}),B,S)$ in Example \ref{sln2} with the  reflection $\tau(x)=-\overline{x}^T$, we have a real Poisson Lie group $\mathrm{SL}(n,\mathbb{C})$ and a Poisson homogeneous space $\mathrm{SL}(n,\mathbb{C})/\mathrm{SU}(n)$, where the latter is actually the dual Poisson Lie group $\mathrm{SU}(n)^*=\mathrm{SB}(n)$, consisting of upper triangular matrices  with real diagonals in $\mathrm{SL}(n,\mathbb{C})$.\end{ex}

\begin{ex}
For the quadratic Rota-Baxter Lie algebra $(\mathrm{sl}(2,\mathbb{C}),B,S)$ in Example \ref{invariant} with the reflection $\tau(x)=-x^T$, we have a complex Poisson Lie group
$\mathrm{SL}(2,\mathbb{C})$ and a Poisson homogeneous space $\mathrm{SL}(2,\mathbb{C})/\mathrm{SO}(2,\mathbb{C})$, consisting of symmetric matrices in $\mathrm{SL}(2,\mathbb{C})$. This Poisson homogeneous space is the fixed locus of the Poisson involution $\mathrm{SL}(2,\mathbb{C})\to \mathrm{SL}(2,\mathbb{C}), x\mapsto x^T$ studied in \cite{Xu}.
\end{ex}

For simplicity, we denote by $\rho$ both the Lie group action of $G$ on $V$ and its infinitesimal Lie algebra action of $\g$ on $V$. Associated to such an action, we have the semi-direct product Lie group $G\ltimes_\rho V$ with the group multiplication $(g,x)(h,y)=(gh,x+\rho_g y)$.

\begin{thm}\label{o}
Let $G$ be a Lie group with a representation $\rho: G\to \mathrm{GL}(V)$ on $V$. Suppose that  $(V\xrightarrow{T} \g,\rho)$ is a  relative Rota-Baxter Lie algebra. Then
\begin{itemize}
\item[\rm (i)] we have a Poisson Lie group  $(G\ltimes_{\rho^*} V^*, \pi)$, where
\[\pi_{(g,\alpha)}\in (T_g G\wedge V^*)\oplus \wedge^2 V^*\subset \wedge^2 T_{(g,\alpha)} (G\ltimes_{\rho^*}V^*)\] is given by
\begin{eqnarray*}
\pi_{(g,\alpha)}&=&(\rho^*_g\otimes \huaL_{g})T-(\huaL_{g}\otimes \rho^*_g) T^{21}-(\iii_{V^*}\otimes \huaR_{g})T+(\huaR_{g}\otimes \iii_{V^*})T^{21}\\ &&-(\iii_{V^*}\otimes f_{\rho,\alpha})T+ (f_{\rho,\alpha}\otimes \iii_{V^*})T^{21},
\end{eqnarray*}
where $f_{\rho,\alpha}: \g\to V$ is $f_{\rho,\alpha}(x):=\rho^*_x(\alpha)$ and we treat $T\in V^*\otimes \g$. Note that $V^*$ with the linear Poisson structure coming from the descendent Lie algebra $V_T$ is a Poisson subgroup of $G\ltimes_{\rho^*} V^*$. \item[\rm (ii)]
 suppose $(\sigma,\tau)$ is a reflection on $(V\xrightarrow{T} \g,\rho)$ as defined in Definition \ref{refzero}.
Then we have a  Poisson homogeneous space
\[\big((G\ltimes_{\rho^*} V^*)/(H\ltimes_{\rho^*} \ker(\tau^*+\iii_{V^*})),\mathrm{pr}_*\pi\big),\]
of the Poisson Lie group $(G\ltimes_{\rho^*} V^*, \pi)$, where $H$ is the Lie group of the Lie algebra $\h=\ker(\sigma-\iii_\g)\subset \g$.
\end{itemize}
\end{thm}
\begin{proof}
For (i), recall that there is a  triangular Lie bialgebra $(\g\ltimes_{\rho^*} V^*, r_T=T-T^{21})$ associated to a relative Rota-Baxter Lie algebra $(V\xrightarrow{T} \g,\rho)$. Its corresponding Poisson Lie group is $(G\ltimes_{\rho^*} V^*, \pi=\overleftarrow{r_T}-\overrightarrow{r_T})$.
Now we compute the left and right translations of the Lie group $G\ltimes_{\rho^*} V^*$. For $x\in \g$ and $\beta\in V^*$, we have
\begin{eqnarray*}
\huaL_{(g,\alpha)} x&=&\frac{d}{dt}|_{t=0}(g,\alpha)(\exp^{tx},0)=\frac{d}{dt}|_{t=0}(g\exp^{tx},\alpha)=\huaL_{g}x;\\
\huaL_{(g,\alpha)} \beta&=&\frac{d}{dt}|_{t=0}(g,\alpha)(e,t\beta)=\frac{d}{dt}|_{t=0}(g,\alpha+t\rho^*_g\beta)=\rho^*_g\beta;\\
\huaR_{(g,\alpha)} x&=&\frac{d}{dt}|_{t=0}(\exp^{tx},0)(g,\alpha)=\frac{d}{dt}|_{t=0}(\exp^{tx}g,\rho^*_{\exp^{tx}}\alpha)=(\huaR_{g}x,\rho^*_x \alpha);\\
\huaR_{(g,\alpha)} \beta&=&\frac{d}{dt}|_{t=0}(e,t\beta)(g,\alpha)=\frac{d}{dt}|_{t=0}(g,t\beta+\alpha)=\beta.
\end{eqnarray*}
Therefore, by calculating $\overleftarrow{r_T}-\overrightarrow{r_T}$, we obtain the desired formula of $\pi_{(g,\alpha)}$.
Observing  that $\huaL_e=\huaR_e=\iii_\g$ and $\rho_e^*=\iii_{V^*}$, we see \[\pi_{(e,\alpha)}=-(\iii_{V^*}\otimes f_{\rho,\alpha})T+ (f_{\rho,\alpha}\otimes \iii_{V^*})T^{21}\in \wedge^2 V^*,\quad \forall \alpha\in V^*.\]
Acting on linear functions $u,v\in V=C^\infty_{\mathrm{lin}}(V^*)$, we have
\[\{u,v\}(\alpha)=\pi_{(e,\alpha)}(du,dv)=\langle \alpha, \rho(T(u))v-\rho(T(v))u\rangle=\langle \alpha,[u,v]_T\rangle.\]
This shows that $V^*$ is a Poisson subgroup of $G\ltimes_{\rho^*} V^*$ and the induced Poisson structure on $V^*$ is  the linear Poisson structure from the descendent Lie algebra $V_T$.

For (ii), first note that $H\ltimes_{\rho^*} \ker(\tau^*+\iii_{V^*})\subset G\ltimes_{\rho^*} V^*$ is a Lie subgroup with Lie subalgebra $\h\ltimes_{\rho^*} \ker(\tau^*+\iii_{V^*})=\ker (\sigma\oplus (-\tau^*)-\iii_{\g\ltimes_{\rho^*} V^*})$. It is known that the quotient space $(G\ltimes_{\rho^*} V^*)/(H\ltimes_{\rho^*} \ker(\tau^*+\iii_{V^*}))$ is a Poisson homogeneous space such that the projection \[\mathrm{pr}:  G\ltimes_{\rho^*} V^*\to (G\ltimes_{\rho^*} V^*)/(H\ltimes_{\rho^*} \ker(\tau^*+\iii_{V^*}))\]
is a Poisson map if and only if the annihilator space
\[(\h\ltimes_{\rho^*} \ker(\tau^*+\iii_{V^*}))^\perp=\h^\perp \oplus \mathrm{Im}(\tau+\iii_{V}) \subset  (\g^*\oplus V)_{r_T}\]
is a Lie subalgebra of the dual Lie algebra. By Theorem \ref{perpzero}, we obtain the result.
\end{proof}
As a consequence, we obtain a Poisson algebra structure on the function space $C^\infty(G\ltimes_{\rho^*} V^*)$. It has a Poisson subalgebra $C^\infty(M)\rtimes \mathrm{Sym}(V)$, where $\mathrm{Sym}(V)$ is the symmetric algebra on $V$.

\begin{cor}
Let $G$ be a Lie group with a representation $\rho: G\to \mathrm{GL}(V)$ on $V$ and let $(V\xrightarrow{T} \g,\rho)$ be a  relative Rota-Baxter Lie algebra. Then we have a Poisson algebra $\big(C^\infty(G)\rtimes \mathrm{Sym}(V), \{\cdot,\cdot\}\big)$, where $\{\cdot,\cdot\}$ is determined by
\begin{eqnarray*}
\{F,u\}(g)
&=&\frac{d}{dt}|_{t=0}F\big(\exp^{-tT(\rho_g u)}g\exp^{tT(\rho_g u)}\big), \quad g\in G,\\ \{u,v\}&=&\rho(Tu)v-\rho(Tv)u,\\
\{F,F'\}&=&0,\end{eqnarray*}
for $F,F'\in C^\infty(G)$ and $u,v\in V$.
\end{cor}

If $\sigma:\g\to \g$ is a Lie algebra automorphism such that $\rho(\sigma x)v=\rho(x)v$ and $\sigma\circ T=T$, then $(\sigma,\iii_{V})$ is a reflection on the relative Rota-Baxter Lie algebra $(V\xrightarrow{T} \g, \rho)$. Denote by $H$ the Lie group of the Lie algebra $\h=\ker (\sigma-\iii_\g)$. Then the above Poisson algebra has a Poisson subalgebra $C^\infty(G)^H\rtimes \mathrm{Sym}(V)^H$ on the $H$-invariant functions.

\begin{cor}
Let $G$ be a connected and simply-connected Lie group such that its Lie algebra $(\g,B)$ is a Rota-Baxter Lie algebra of weight $\lambda$. Then
\begin{itemize}
\item[\rm (i)] $(G\ltimes_{\Ad^*} \g^*,\pi)$ is a Poisson Lie group, where $\pi_{(g,\alpha)}\in \wedge^2 T_{(g,\alpha)}(G\ltimes_{\Ad^*} \g^*)$ is defined by
\begin{eqnarray*}
\pi_{(g,\alpha)}&=&(\Ad_{g}^*\otimes \huaL_{g})B-(\huaL_{g}\otimes \Ad_{g}^*) B^{21}-(\iii_{\g^*}\otimes \huaR_{g})B+(\huaR_{g}\otimes \iii_{\g^*})B^{21}\\ &&-(\iii_{\g^*}\otimes f_{\ad,\alpha})B+ (f_{\ad,\alpha}\otimes \iii_{\g^*})B^{21},\quad \lambda=0;\\
\pi_{(g,\alpha)}&=&\frac{1}{2\lambda}\big((\Ad_{g}^*\otimes \huaL_{g})B-(\huaL_{g}\otimes \Ad_{g}^*) B^{21}-(\iii_{\g^*}\otimes \huaR_{g})B+(\huaR_{g}\otimes \iii_{\g^*})B^{21}\\ &&-(\iii_{\g^*}\otimes f_{\ad,\alpha})B+ (f_{\ad,\alpha}\otimes \iii_{\g^*})B^{21}\big),\quad \lambda\neq 0,
\end{eqnarray*}
where  $B:\g\to \g$ is viewed as an element in $\g^*\otimes \g$ with $B^{21}\in \g\otimes \g^*$ and $f_{\ad,\alpha}: \g\to \g^*$ is $f_{\ad,\alpha}(x):=\ad^*_x(\alpha)$.\item[\rm (ii)] if $\tau:\g\to \g$ is a reflection on the Rota-Baxter Lie algebra as introduced in Definition \ref{reflectiononRB}, then
\[\big((G\ltimes_{\Ad^*} \g^*)/(H\ltimes_{\Ad^*} \ker(\tau^*+\iii_{\g^*})), \mathrm{pr}_* \pi\big)\] is a Poisson homogeneous space of the Poisson Lie group $(G\ltimes_{\Ad^*} \g^*,\pi)$, where $H\subset G$ is the closed Lie subgroup with Lie algebra $\h=\ker(\tau-\iii_{\g})$.
\end{itemize}
\end{cor}

\begin{proof}
As a Rota-Baxter Lie algebra $(\g,B)$ gives rise to a quadratic Rota-Baxter Lie algebra $(\g\ltimes_{\ad^*} \g^*, B\oplus (-\lambda\iii_{\g^*}-B^*), \huaS)$, this result is a consequence of Proposition \ref{quagroup}. Based on a similar computation as in the proof of Proposition \ref{o}, we obtain the explicit formula of $\pi$.
\end{proof}

\begin{ex}
Let $G$ be a Lie group with Lie algebra $\g$. The Rota-Baxter Lie algebra $(\g,-\iii_\g)$ induces a Poisson Lie group
$(G\ltimes_{\Ad^*}\g^*, \pi)$, where $\pi_{(g,\alpha)}=-\pi_{\g^*}|_\alpha$ and $\pi_{\g^*}$ is the linear Poisson structure on $\g^*$.  Moreover,  any involutive Lie algebra automorphism $\tau: \g\to \g$ gives a Poisson homogeneous space
$\big((G\ltimes_{\Ad^*}\g^*)/(H\ltimes_{\Ad^*} \h^\perp),\mathrm{pr}_*\pi\big)$.
\end{ex}

 \begin{pro}
Suppose that $G$ is a Lie group and  $\rho: G\to \Aut(K)$ is an action of $G$ on another Lie group $K$ by automorphisms. Let $(\frkk\xrightarrow{T} \g,\rho)$ be a relative Rota-Baxter Lie algebra. Then we have a coboundary Poisson Lie group
\[\big((G\ltimes_{\rho} K)\ltimes_{\Ad^*} (\g^*\oplus \frkk^*), \overleftarrow{r_{\overline{T}}}-\overrightarrow{r_{\overline{T}}}\big),\]
which has a Poisson subgroup
\[\big(G\ltimes \frkk^*,(\overleftarrow{r_{\overline{T}}}-\overrightarrow{r_{\overline{T}}})|_{G\ltimes \frkk^*}\big),\]
and a Poisson homogeneous space
\[\big(K\ltimes \g^*,\mathrm{pr}_*(\overleftarrow{r_{\overline{T}}}-\overrightarrow{r_{\overline{T}}})\big),\]
where $\mathrm{pr}: (G\ltimes_{\rho} K)\ltimes (\g^*\oplus \frkk^*)\to K\ltimes \g^*$ is the natural projection.
\end{pro}

\begin{proof}
From Proposition  \ref{main}, we see that $\g\ltimes_{\rho^*} \frkk^*\subset (\g\ltimes_{\rho} \frkk)\ltimes_{\ad^*} (\g^*\oplus \frkk^*)$ is a Lie subalgebra and \[(\g\ltimes_{\rho^*} \frkk^*)^\perp=\g\oplus \frkk^*\subset (\g\oplus \frkk \oplus \g^*\oplus \frkk^*)_{r_{\overline{T}}}\] is an ideal. Thus $G\ltimes \frkk^*$ is a Poisson subgroup. Note that
$r_{\overline{T}}(\g^*\oplus  \frkk,\g^*\oplus \frkk)=0$. The Poisson structure $(\overleftarrow{r_{\overline{T}}}-\overrightarrow{r_{\overline{T}}})|_{G\ltimes \frkk^*}$ is in general not coboundary.

Unlike the above situation, $\frkk\ltimes \g^*\subset (\g\ltimes_{\rho} \frkk)\ltimes_{\ad^*} (\g^*\oplus \frkk^*)$ is a Lie subalgebra, and its annihilator space \[(\frkk\ltimes \g^*)^\perp=\frkk \oplus \g^*\subset (\g\oplus \frkk \oplus \g^*\oplus \frkk^*)_{r_{\overline{T}}}\] is just a Lie subalgebra, not an ideal, since $[u, y]=[Tu,y]\in \g$ for $u\in \frkk$ and $y\in \g$. Hence $K\ltimes \g^*=(G\ltimes_{\rho} K)\ltimes_{\Ad^*} (\g^*\oplus \frkk^*)/G\ltimes \frkk^*$ with the quotient Poisson structure is a Poisson homogeneous space.
\end{proof}
By \cite{ LM}, for two manifolds $M_1$ and $M_2$, a {\bf mixed product Poisson structure} on the product manifold $M_1\times M_2$
is a Poisson bivector field on $M_1\times M_2$ that projects to well-defined Poisson structures on $M_1$ and $M_2$.  The above Poisson structure on  $(G\ltimes_{\rho} K)\ltimes_{\Ad^*} (\g^*\oplus \frkk^*)$ is a mixed product Poisson structure on the product manifold $(G\ltimes \frkk^*)\times (K\ltimes \g^*)$.

\begin{cor}
If the pair $\tau:\frkk\to \frkk$ and $\sigma: \g\to \g$ is a reflection on the relative Rota-Baxter Lie algebra $(\frkk\xrightarrow{T}\g, \rho)$ as introduced in Definition \ref{refrel}, then
\[\big(((G\ltimes_{\rho} K)\ltimes_{\Ad^*} (\g^*\oplus \frkk^*))/(H\ltimes_{\Ad^*} \ker(\sigma^*\oplus\tau^*+\iii_{\g^*\oplus \frkk^*})), \mathrm{pr}_* \pi\big)\] is a Poisson homogeneous space of the Poisson Lie group $\big((G\ltimes_{\rho} K)\ltimes_{\Ad^*} (\g^*\oplus \frkk^*),\pi\big)$, where $H\subset G\ltimes K$ is the closed Lie subgroup with Lie algebra $\h=\ker(\sigma\oplus\tau-\iii_{\g\oplus\frkk})$.
\end{cor}

\end{document}